\documentclass[a4paper]{article}

\usepackage[lofdepth,lotdepth]{subfig}
\usepackage[english]{babel}
\usepackage[T1]{fontenc}
\usepackage[utf8]{inputenc}
\usepackage{subfig}
\usepackage{multirow}
\usepackage[font = scriptsize, labelfont = bf, labelsep = period, tableposition = top, singlelinecheck = false, justification = centering]{caption}
\usepackage{cite}
\usepackage{placeins}
\usepackage{amssymb}
\usepackage{amsmath}
\usepackage{amsthm}
\usepackage{latexsym}
\usepackage{wrapfig}
\usepackage{float}
\usepackage{mathtools}
\usepackage{bm}
\usepackage{listings}
\usepackage{textcomp}
\usepackage{setspace}
\usepackage[dvipsnames]{xcolor}
\usepackage{etoolbox}
\usepackage{hyperref}

\usepackage{algorithm,algorithmic}
\usepackage{authblk}

\newtheorem{theo}{Theorem}
\newtheorem{prop}{Proposition}
\newtheorem{cor}{Corollary}
\newtheorem{lem}{Lemma}
\theoremstyle{definition}

\newtheorem{exa}{Example}

\newtheorem{rem}{Remark}

\usepackage{float}
\newfloat{algorithm}{t}{lop}

\newcommand{\U}{{\mathbb U}}
\newcommand{\F}{{\mathbb F}}

\newcommand{\Z}{{\mathbb Z}}

\newcommand{\C}{{\mathbb C}}

\newcommand{\cC}{{\mathcal C}}

\newcommand{\cG}{{\mathcal G}}
\newcommand{\cH}{{\mathcal H}}

\newcommand{\cP}{{\mathcal P}}
\newcommand{\cO}{{\mathcal O}}

\newcommand{\cL}{{\mathcal L}}
\newcommand{\cI}{{\mathcal I}}
\newcommand{\cJ}{{\mathcal J}}

\renewcommand{\cH}{{\mathcal H}}
\newcommand{\cS}{{\mathcal S}}
\newcommand{\cN}{{\mathcal N}}

\newcommand{\cV}{{\mathcal V}}
\newcommand{\cW}{{\mathcal W}}
\newcommand{\cX}{{\mathcal X}}

\newcommand{\cZ}{{\mathcal Z}}
\newcommand{\bP}{{\mathbf P}}
\newcommand{\bH}{{\mathbf H}}
\newcommand{\bZ}{{\mathbf Z}}

\newcommand{\rr}{{\widetilde{r}}}

\newcommand{\wcI}{\widetilde{\cI}}
\newcommand{\HI}{\bH_{\cI}}
\newcommand{\PI}{\bP_{\scriptsize\cI}}
\newcommand{\gf}{\mbox{$\bG_{\scriptsize{\bF}}$}}
\newcommand{\wb}{\bw_{\scriptsize{{\bf b}}}}
\newcommand{\HW}{\cH\cW_{\!N}}
\newcommand{\Fm}{\F_2^{2m}}
\newcommand{\wHI}{\widetilde{\bH_{\cI}}}

\newcommand{\sbt}{\raisebox{.2ex}{\mbox{$\,\scriptscriptstyle\bullet\,$}}}

\newcommand{\Sp}{\mathrm{Sp}}

\newcommand{\rk}{\mbox{${\rm rank\,}$}}
\newcommand{\sperp}{{{\perp_{\rm{s}}}}}

\newcommand{\cl}{\mbox{${\rm Cliff\!}$}}

\newcommand{\rs}{\mbox{\rm rs}\,}
\newcommand{\cs}{\mbox{\rm cs}\,}

\newcommand{\spann}{\mbox{\rm span}}

\newcommand{\wt}{{\rm wt}}

\newcommand{\chor}{\mbox{${\rm d}_{\rm c}$}}

\newcommand{\T}{\mbox{$^{\sf T}$}}

\newcommand{\inners}[2]{\mbox{$\langle{\,{#1}\,}|\,{#2}\,\rangle_{\rm s}$}}

\newcommand{\ov}[1]{\mbox{$\overline{#1}$}}
\newcommand{\FDP}{\bF_D(\bP)}
\newcommand{\FUS}{\bF_U(\bS)}
\newcommand{\GL}{\mathrm{GL}}
\newcommand{\DG}{\mathrm{DG}}
\newcommand{\Sym}{\mathrm{Sym}}

\newcommand{\diag}{\textup{diag}}

\newcommand{\bb}[1]{\mbox{\rm $\textbf{#1}$}}
\newcommand{\chirp}{\mbox{$\rm{\cV_{\mathrm{BC}}}$}}
\newcommand{\bssc}{\mbox{$\rm{\cV_{\mathrm{BSSC}}}$}}
\newcommand{\Imr}{\bI_{m|r}}
\newcommand{\Imrr}{\bI_{m|-r}}

\newcommand{\bA}{{\mathbf A}}
\newcommand{\bB}{{\mathbf B}}
\newcommand{\bC}{{\mathbf C}}
\newcommand{\bD}{{\mathbf D}}
\newcommand{\bE}{{\mathbf E}}
\newcommand{\bF}{{\mathbf F}}
\newcommand{\bG}{{\mathbf G}}
\newcommand{\bI}{{\mathbf I}}

\newcommand{\bM}{{\mathbf M}}
\newcommand{\bN}{{\mathbf N}}

\newcommand{\bS}{{\mathbf S}}
\newcommand{\bU}{{\mathbf U}}

\newcommand{\bX}{{\mathbf X}}

\newcommand{\ba}{{\mathbf a}}
\newcommand{\bc}{{\mathbf c}}
\newcommand{\bd}{{\mathbf d}}
\newcommand{\be}{{\mathbf e}}
\newcommand{\bn}{{\mathbf n}}

\newcommand{\bs}{{\mathbf s}}

\newcommand{\bu}{{\mathbf u}}
\newcommand{\bv}{{\mathbf v}}
\newcommand{\bw}{{\mathbf w}}
\newcommand{\bx}{{\mathbf x}}
\newcommand{\by}{{\mathbf y}}
\newcommand{\bz}{{\mathbf z}}
\newcommand{\twomat}[2]{\mbox{$\left[\!\!\begin{array}{cc}{#1}&{#2}\end{array}\!\!\right]$}}
\newcommand{\fourmat}[4]{\mbox{$\left[\!\!\begin{array}{cc}{#1}&{#2}\\{#3}&{#4}\end{array}\!\!\right]$}}
\newcommand{\twomatv}[2]{\mbox{$\left[\!\!\begin{array}{c}{#1}\\{#2}\end{array}\!\!\right]$}}

\newcounter{alp}
\newcounter{ara}
\newcounter{rom}

\newif\ifcomment
\commenttrue

\usepackage[margin=1.3 in]{geometry}
\date{}

\begin{document}
\title{Binary Subspace Chirps and their Applications}
\title{Binary Subspace Chirps, Quantum Computation, and Random Access}
\title{\Huge \textbf{Binary Subspace Chirps}}
\author[*]{Tefjol Pllaha}
\author[*]{Olav Tirkkonen}
\author[$\dagger$]{Robert Calderbank}
\affil[*]{\small{Department of Communications and Networking, Aalto University, Finland}}
\affil[$\dagger$]{Department of Electrical and Computer Engineering, Duke University, NC, USA}
\affil[ ]{Emails: \{tefjol.pllaha, olav.tirkkonen\}@aalto.fi, robert.calderbank@duke.edu}

	
    

\maketitle

\begin{abstract}
We describe in details the interplay between binary symplectic geometry and quantum computation, with the ultimate goal of constructing highly structured codebooks. The Binary Chirps (BCs)
are Complex Grassmannian Lines in $N = 2^m$ dimensions used in deterministic compressed sensing and random/unsourced multiple access in wireless networks.
Their entries are fourth roots of unity and can be described in terms of second order Reed-Muller codes.
The Binary Subspace Chirps (BSSCs) are a unique collection of BCs of \emph{ranks} ranging from $r=0$ to $r = m$, embedded in $N$ dimensions according to an on-off pattern determined by a rank $r$ binary subspace. This yields a codebook that is asymptotically 2.38 times larger than the codebook of BCs, has the same minimum chordal distance as the codebook of BCs, and the alphabet is minimally extended from $\{\pm 1,\pm i\}$ to $\{\pm 1,\pm i, 0\}$. 
Equivalently, we show that BSSCs are stabilizer states, and we characterize them as columns of a well-controlled collection of Clifford matrices.
By construction, the BSSCs inherit all the properties of BCs, which in turn makes them good candidates for a variety of applications.
For applications in wireless communication, we use the rich algebraic structure of BSSCs to construct a low complexity decoding algorithm that is reliable against Gaussian noise.
In simulations,  BSSCs exhibit an error probability comparable or slightly lower than BCs, both for single-user and multi-user transmissions.
\end{abstract}


\section{Introduction}

Codebooks of complex projective (Grassmann) lines, or tight frames,
have found application in multiple problems of interest for communications
and information processing, such as code division multiple access
sequence design~\cite{Viswanath1999}, precoding for multi-antenna
transmissions~\cite{Love2003} and network coding~\cite{Kotter2008}.
Contemporary interest in such codes arise, e.g., from deterministic
compressed sensing~\cite{DeVore2007,HCS08,Li2014,Wang2018,TC18}, virtual
full-duplex communication~\cite{GZ10}, mmWave
communication~\cite{Tsai2018}, and random
access~\cite{Calderbank2019}.

One of the challenges/promises of 5G wireless communication is to enable massive machine-type communications (mMTC) in the Internet of Things (IoT), in which a massive number of low-cost devices sporadically and randomly access the network~\cite{Polyanskiy17}. 
In this scenario, users are assigned a unique \emph{signature} sequence which they transmit whenever active~\cite{UEP15}. 
A twin use-case is unsourced multiple access, where a large number of messages is transmitted infrequently.
Polyanskiy~\cite{Polyanskiy17} proposed a framework in which communication occurs in blocks of $N$ channel uses, and the task of a receiver is to identify correctly $L$ active users (messages) out of $2^B$, with one regime of interest being $N = 30, 000, L = 250$, and $B = 100$. 
Ever since its introduction, there have been several follow-up works~\cite{Polyanskiy19,Polyanskiy19-2,10.1007/978-3-030-01168-0_15,Calderbank2019,Narayanan}, extensions to a massive MIMO scenario~\cite{alex2019massive} where the base station has a very large number of antennas, and a discussion on the fundamental limits on what is possible~\cite{kowshik2019fundamental}.

Given the massive number of to-be-supported (to-be-encoded) users (messages), the design criteria are fundamentally different and one simply cannot rely on classical multiple-access channel solutions. 
For instance, interference is unavoidable since it is impossible to have orthogonal signatures/codewords. 
Additionally, given that there is a small number of active user, the interference is limited.
Thus, the challenge becomes to design highly structured codebooks of large cardinality along with a reliable and low-complexity decoding algorithm. \

Codebooks of Binary Chirps (BCs)~\cite{HCS08,AHSC09} provide such highly structured Grassmannian line codebook in $N=2^m$ dimensions with additional desirable properties.
All entries come from a small alphabet, being a fourth root of unity, and can be described in terms of second order Reed-Muller (RM) codes. 
RM codes have the fascinating property that a Walsh-Hadamard measurement cuts in half the solution space. This yields a single-user decoding complexity of ${\cal O}(N\log^2 N)$, coming from the Walsh-Hadamard transform and number of required measurements. 
Additionally, the number of codewords is
reasonably large, growing as $2^{m(m+3)/2} = \sqrt{N}^{3+\log_2 N}$, while the minimum \emph{chordal distance} is $1/\sqrt{2}$.

We expand the BC codebook to the codebook of Binary Subspace Chirps (BSSCs) by collectively considering all BCs in $S = 2^r$ dimensions, $r = 0,\ldots,m$, in $N = 2^m$ dimensions. 
That is, given a BC in $S = 2^r$ dimensions, we embed it in $N = 2^m$ dimensions via a unique \emph{on-off pattern} determined by a rank $r$ binary subspace. 
Thus, a BSSC is characterized by a \emph{sparsity} $r$, a BC part parametrized by a binary symmetric matrix $\bS_r \in \Sym(r;2)$ and a binary vector ${\bf b}\in \F_2^m$, and a unique on-off pattern parametrized by a rank $r$ binary subspace $H\in \cG(m,r;2)$; see~\eqref{e-bssc} for the formal definition. 
The codebook of BSSCs inherits all the desirable properties of BCs, and in addition, it has asymptotically about 2.384 more codewords. 
Thus, an active device with a rank $r$ signature will transmit $\alpha/\sqrt{2^r}$, $\alpha\in \{\pm 1,\pm i\}$ during time slots determined by the rank $r$ subspace $H$, and it will be silent otherwise. 
This resembles the model of~\cite{GZ10}, in which active devices can also be used (to listen) as receivers during the off-slots.

Given the structure of BSSCs, a unified rank, on-off pattern, and BCs part (in this order) estimation technique is needed. In~\cite{TC19}, a reliable on-off pattern detection was proposed, which made use of a Weyl-type transform~\cite{QTCS16} on $m$ qubit diagonal Pauli matrices; see~\eqref{e-abs}. The algorithm can be described with the common language of symplectic geometry and quantum computation. 
The key insight here is to view BSSCs as common eigenvectors of maximal sets of commuting Pauli matrices, commonly referred as \emph{stabilizer groups}. 
Indeed, we show that BSSCs are nothing else but \emph{stabilizer states}~\cite{DM03}, and their sparsity is determined by the diagonal portion of the corresponding stabilizer group; see Corollaries~\ref{C-states} and~\ref{C-SZN}.
We also show that each BSSC is a column of a unique Clifford matrix~\eqref{e-eqq2}, which itself is the common eigenspace of a unique stabilizer group~\eqref{e-ES1}; see also Theorem~\ref{T-ES}. The interplay between the binary world and the complex world is depicted in Figure~\ref{Fig-Struct}.
\begin{figure}
\begin{center}
\includegraphics[width = 0.45\textwidth]{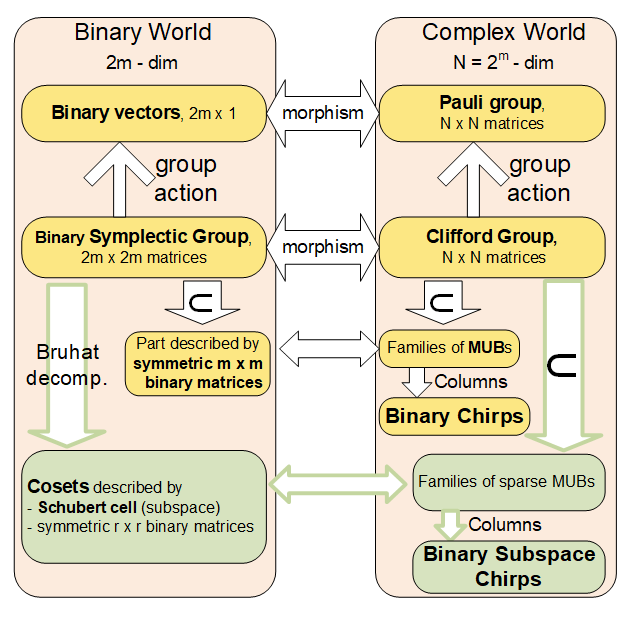}
\caption{Interplay of binary world and complex world. Prior art is depicted in yellow. The contributions of this paper are depicted in green. See also~\cite{PTC20,TC19}.} 
\label{Fig-Struct}
\end{center}
\end{figure}

Making use of these structural results, the on-off pattern detection of~\cite{TC19} can be generalized to recover the BC part of the BSSC, this time by using the Weyl-type transform on the off-diagonal part of the corresponding stabilizer group. 
This yields a single-user BSSC reconstruction as described in Algorithm~\ref{alg}. 
In~\cite{PTC20}, we added Orthogonal Matching Pursuit (OMP) to obtain a multi-user BSSCs reconstruction (see Algorithm~\ref{alg1}) with reliable performance when there is a small number of active users. 
As the number of active users increases, so does the interference, which has a quite destructive effect on the on-off pattern. However, state-of-the-art solutions for BCs~\cite{TC18,Narayanan,massiveGuo} such as slotting and patching, can be used to reduce the interference. 
Preliminary simulations show that BSSCs exhibit a lower error probability than BCs. This is because BSSCs have fewer closest neighbors on average than BCs. In addition, BSSCs are uniformly distributed over the sphere, which makes them optimal when dealing with Gaussian noise.

Throughout, the decoding complexity is kept at bay from the underlying symplectic geometry. 
The sparsity, the BC part, and the on-off pattern of a BSSC can be described in terms of the Bruhat decomposition~\eqref{e-Bruhat1} of a symplectic matrix. 
Indeed, the unique Clifford matrix~\eqref{e-eqq2} of which a BSSC is a column, is parametrized by a coset representative~\eqref{e-generic} as described in Lemma~\ref{L-RightCoset}. 
In turn, such coset representative determines a unique stabilizer group~\eqref{e-ES1}. We use this interplay to reconstruct a BSSC by reconstructing the stabilizer group that stabilizes the given BSSC. 
This alone reduces the complexity from $\cO(N^2)$ to $\cO(N\log_2 N)$.

The paper is organized as follows.
In Section~\ref{Sec-Preliminaries} we review the basics of binary symplectic geometry and quantum computation. In order to obtain a unique parametrization of BSSCs, we use Schubert cells and the Bruhat decomposition of the symplectic group. 
In Section~\ref{Sec-Cliff} we lift the Bruhat decomposition of the symplectic group to obtain a decomposition of the Clifford group. Additionally, we parametrize those Clifford matrices whose columns are BSSCs.
In Section~\ref{Sec-BSSC} we give the formal definition of BSSCs, along with their algebraic and geometric properties. 
In Sections~\ref{Sec-Recon} and~\ref{Sec-multiBSSC} we present reliable low complexity decoding algorithms, and discuss simulation results. We end the paper with some conclusions and directions future research.

\subsection{Conventions}
All vectors, binary or complex, will be columns. 
$\F_2$ denotes the binary field, $\GL(m;2)$ denotes the group of binary $m\times m$ invertible matrices, and $\Sym(m;2)$ denotes the group of binary $m\times m$ symmetric matrices. We will denote matrices (resp., vectors) with upper case (resp., lower case) bold letters. 
$\bA\T$ will denote the transpose and $\bA^{\!-\sf T}$ will denote the inverse transposed. 
$\cs (\bA)$ and $\rs (\bA)$ will denote the column space and the row space of $\bA$ respectively. 
Since all our vectors are columns, we will typically deal with column spaces, except when we work with notions from quantum computation, where row spaces are customary. 
$\bI_m$ will denote the $m\times m$ matrix (complex or binary). $\cG(m,r;2) \cong \GL(m;2)/\GL(r;2)$ denotes the binary
Grassmannian, that is, the set of all $r$-dimensional subspaces of $\F_2^m$. 
$\U(N)$ denotes the set of unitary $N\times N$ complex matrices and $\dagger$ will denote the conjugate transpose of a matrix.


\section{Preliminaries} \label{Sec-Preliminaries}
In this section we will introduce all preliminary notions needed for navigating the connection between the $2m$ dimensional binary world and the $2^m$ dimensional complex world, as depicted in Figure~\ref{Fig-Struct}. 
The primary bridge used here is the well-known homomorphism~\eqref{e-Phi} and the Bruhat decomposition of the symplectic group. We focus on cosets of the symplectic group modulo the semidirect product $\GL(m;2)\rtimes \Sym(m;2)$. 
These cosets are characterized by a rank $r=0,\ldots,m$ and a binary subspace $H\in \cG(m,r;2)$, which we will think of as the column space of an $m\times r$ binary matrix in column reduced echelon form. 
We will use Schubert cells as a formal and systematic approach. This also provides a framework for describing well-known facts from binary symplectic geometry (e.g., Remark~\ref{R-states}). Finally, Subsection~\ref{Sec-HW} discusses common notions from quantum computation.
\subsection{Schubert Cells}\label{S-SC}
Here we discuss the Schubert decomposition of the Grassmannian $\cG(m,r;2)$ with the respect to the standard \emph{flag} 
\begin{equation}
    \{\mathbf{0}\} = V_0 \subset V_1 \subset \cdots \subset V_m,
\end{equation}
where $V_i = \spann\{\be_1,\ldots, \be_i\}$ and $\{\be_1,\ldots, \be_m\}$ is the standard basis of $\F_2^m$. 
Fix a set of $r$ indices $\cI = \{i_1,\ldots, i_r\} \subset \{1,\ldots, m\}$, which, without loss of generality, we assume to be in increasing order. 
The \emph{Schubert cell} $\cC_{\cI}$ is the set of all $m\times r$ matrices that have $1$ in \emph{leading} positions $(i_j,j)$, $0$ on the left, right, and above each leading position, and every other entry is free. This is simply the set of all binary matrices in \emph{column reduced echelon form} with leading positions $\cI$.
By counting the number of free entries in each column one concludes that 
\begin{equation}
    \dim \cC_{\cI} = \sum_{j=1}^r (m - i_j) - (r - j).
\end{equation}

Fix $H \in \cG(m,r;2)$, and think of it as the column space of a $m\times r$ matrix $\bH$. After column operations, it will belong to some cell $\cC_{\cI}$, and to emphasize this fact, we will denote it as $\HI$. 
Schubert cells have a well-known duality theory which we outline next. 
Let $\widetilde{\bH_{\cI}}$ be such that $(\HI)\T \wHI = 0$. 
Of course $\cs (\wHI)\in \cG(m,m-r;2)$. 
Let $\widetilde{\cI}:= \{1,\ldots,m\}\setminus \cI$ and put $\widetilde{\cC_{\cI}}:= \{\wHI\mid \HI \in \cC_{\cI}\}$.
There is a bijection between $\{\widetilde{\cC_{\cI}}\}_{|\cI| = r}$ and $\{\cC_{\widetilde{\cJ}}\}_{|\cJ| = r}$, realized by reverting the rows and columns of $\wHI$ and identifying $\widetilde{\cJ} = \{i_1,\ldots,i_{m-r}\}$ with $\widehat{\cJ}:=\{m-i_{1}+1,\ldots,m-i_{m-r}+1\}$.
With this identification, we will denote $\bH_{\widetilde{\cI}}$ the unique element of cell $\cC_{\widetilde{\cI}}$ that is equivalent with $\wHI$, obtained by reverting the rows and columns of $\wHI$: 
\begin{equation}
\bH_{\widetilde{\cI}} = \bP_{{\rm ad},m}\wHI \bP_{{\rm ad},m-r},
\end{equation}
where $\bP_{\rm ad}$ is the antidiagonal matrix in respective dimensions.

Each cell has a distinguished element: $\bI_{\cI}\in \cC_{\cI}$ will denote the identity matrix $\bI_m$ restricted to $\cI$, that is, the unique element in $\cC_\cI$ that has all the free entries 0. 
Note that $\bI_{\cI}$ has as $j$th column the $i_j$th column of $\bI_m$, and thus its non-zero rows form $\bI_r$. In particular if $|\cI| =m$ then $\bI_{\cI} = \bI_m$. 
We also have $\bI_{\widetilde{\cI}}\in \cC_{\wcI}$. With this notation one easily verifies that 
\begin{equation}
\begin{array}{ccc}\label{e-PI1}
(\bI_{\cI})\T\HI = \bI_r, & (\bI_{\cI})\T\bI_{\wcI} = 0, & (\wHI)\T\,\bI_{\wcI} = \bI_{m-r}.\\
\end{array}
\end{equation} 
In addition, $\HI$ can be completed to an invertible matrix
\begin{equation}\label{e-PI}
    \PI:=\twomat{\HI}{\bI_{\wcI}}\hspace{.0001 in}\in \GL(m;2).
\end{equation}
Note that when $\bI_{\cI}$ is completed to an invertible matrix it gives rise to a permutation matrix.
Next,~\eqref{e-PI1} along with the default equality  $(\HI)\T \wHI = 0$ implies that
\begin{equation}\label{e-PIT}
    \PI^{-\sf T} = \twomat{\bI_{\cI}}{\widetilde{\HI}}.
\end{equation}

Let us describe this framework with an example.
\begin{exa}\label{OEx}
Let $m = 3$ and $r = 2$. Then
\begin{align*}
\cC_{\{1,2\}} = \left[\!\!\begin{array}{cc}1&0\\0&1\\u&v\end{array}\!\!\right]\!, \quad & \widetilde{\cC_{\{1,2\}}} = \left[\!\!\begin{array}{c}u\\v\\1\end{array}\!\!\right]\!, \quad \cC_{\widehat{\{3\}}} \cong \cC_{\{1\}}= \left[\!\!\begin{array}{c}1\\v\\u\end{array}\!\!\right]\!, \\ 
\cC_{\{1,3\}} = \left[\!\!\begin{array}{cc}1&0\\u&0\\0&1\end{array}\!\!\right]\!, \quad & \widetilde{\cC_{\{1,3\}}} = \left[\!\!\begin{array}{c}u\\1\\0\end{array}\!\!\right]\!, \quad \cC_{\widehat{\{2\}}} \cong \cC_{\{2\}}= \left[\!\!\begin{array}{c}0\\1\\u\end{array}\!\!\right]\!, \\ 
\cC_{\{2,3\}} = \!\left[\!\!\begin{array}{cc}0&0\\1&0\\0&1\end{array}\!\!\right]\!, \quad & \widetilde{\cC_{\{2,3\}}} = \left[\!\!\begin{array}{c}1\\0\\0\end{array}\!\!\right]\!, \quad \cC_{\widehat{\{1\}}} \cong \cC_{\{3\}}= \left[\!\!\begin{array}{c}0\\0\\1\end{array}\!\!\right] \!.
\end{align*}
Let us focus on $\cI = \{1,3\}$. Then $\cC_\cI$ is constructed directly by definition, that is, in column reduced echelon form with leading positions $1$ and $3$.
Whereas, $\widetilde{\cC_\cI}$ is constructed so that $(\HI)\T\wHI = 0$. Then we revert the rows and columns (only rows in this case) to obtain the last object where we identify\footnote{In this specific case there is no need for identification, but this is only a coincidence. For different choices of $\cI$ one needs a true identification.}  $\{2\} = \widetilde{\{1,3\}}$ with $\widehat{\{2\}} \equiv \{2\}$.

In this case, as we see from above, there is only one free bit. This yields two subspaces/matrices $\bH_\cI$, which when completed to an invertible matrix as in~\eqref{e-PI} yield
\begin{align}\label{e-ex}
    \bP_{u=0} = \left[\!\!\begin{array}{ccc}1&0&0\\0&0&1\\0&1&0\end{array}\!\!\right],\quad \bP_{u=1} = \left[\!\!\begin{array}{ccc}1&0&0\\1&0&1\\0&1&0\end{array}\!\!\right].
\end{align}
Then one directly computes
\begin{align}\label{e-ex1}
    \bP^{-\sf T}_{u=0} = \left[\!\!\begin{array}{ccc}1&0&0\\0&0&1\\0&1&0\end{array}\!\!\right],\quad \bP^{-\sf T}_{u=1} = \left[\!\!\begin{array}{ccc}1&0&1\\0&0&1\\0&1&0\end{array}\!\!\right].
\end{align}
Compare~\eqref{e-ex1} with~\eqref{e-PIT}; the first two columns are obviously $\bI_\cI$, whereas the last column is precisely $\cC_{\widehat{\{2\}}} \equiv \cC_{\{2\}}$ with rows reverted.
Note here that when \emph{all} the free bits are zero then the resulting $\bP$ is simply a permutation matrix, and in this case $\bP^{-\sf T} = \bP$.

\end{exa}

\subsection{Bruhat Decomposition of the Symplectic Group}
We first briefly describe the symplectic structure of $\Fm$ via the symplectic bilinear form
\begin{equation}\label{e-sinner}
    \inners{\ba,{\bf b}}{\bc,\bd}:= {\bf b}\T\bc+\ba\T\bd.
\end{equation}
One is naturally interested in automorphisms that preserve such symplectic structure. It follows directly by the definition that a $2m\times 2m$ matrix $\bF$ preserves $\inners{\sbt}{\sbt}$ iff $\bF\bb{$\mathbf{\Omega}$}\bF\T = \bb{$\mathbf{\Omega}$}$ where 
\begin{equation}
    \bb{$\mathbf{\Omega}$} = \fourmat{{\bf 0}_m}{\bI_m}{\bI_m}{{\bf 0}_m}.
\end{equation}
We will denote the group of all such \emph{symplectic matrices} $\bF$ with $\Sp(2m;2)$. Equivalently, 
\begin{equation}\label{e-bf}
    \bF = \fourmat{\bA}{\bB}{\bC}{\bD} \in \Sp(2m;2)
\end{equation}
iff $\bA\bB\T, \bC\bD\T\in \Sym(m;2)$ and $\bA\bD\T + \bB\bC\T = \bI_m$. It is well-known that
\begin{equation}\label{e-sp}
    |\Sp(2m;2)| = 2^{m^2}\prod_{i=1}^m(4^i - 1).
\end{equation}


Consider the row space $H:=\rs [\bA\,\,\,\,\bB]$ of the $m\times 2m$ upper half of a symplectic matrix $\bF\in \Sp(2m;2)$. 
Because $\bA\bB\T$ is symmetric one has $[\bA\,\,\,\,\bB]\mathbf{\Omega} [\bA\,\,{\bf B}]\T = {\bf 0}$ and thus $\inners{\bx}{\by} = 0$ for all $\bx,\by\in H$. 
We will denote $(\sbt)^\sperp$ the dual with respect to the symplectic inner product~\eqref{e-sinner}.
It follows that $H\subseteq H^\perp$, that is, $H$ is self-orthogonal or \emph{totally isotropic}. 
Moreover, $ H$ is \emph{maximal} totally isotropic because $\dim H = m$ and thus $H = H^\perp$. 
The set of all self-dual/maximal totally isotropic subspaces is commonly referred as the \emph{Lagrangian Grassmannian} $\cL(2m,m;2) \subset \cG(2m,m;2)$. It is well-known that
\begin{equation}\label{e-Lan}
    |\cL(2m,m)| = \prod_{i=1}^m(2^i + 1).
\end{equation}

For reasons that will become clear latter on we are interested in decomposing symplectic matrices into more elementary symplectic matrices, and we will do this via the \emph{Bruhat decomposition} of $\Sp(2m;2)$. While the decomposition holds in a general group-theoretic setting~\cite{Bou68}, here we give a rather elementary approach; see also~\cite{Rao93}. We start the decomposition by writing 
\begin{equation}
    \Sp(2m;2) = \bigcup_{r=0}^m\cC_r,
\end{equation}
where 
\begin{equation}
    \cC_r = \left\{\bF = \fourmat{\bA}{\bB}{\bC}{\bD} \in \Sp(2m;2) \,\, \middle|\,\, \rk \bC = r\right\}.
\end{equation}
In $\Sp(2m;2)$ there are two distinguished subgroups:
\begin{align}
    S_D & := \left\{\FDP = \fourmat{\bP}{{\bf 0}}{{\bf 0}}{\bP^{-\sf T}}\,\,\middle| \,\,\bP\in \GL(m;2)\right\}, \\ 
    S_U & :=\left\{\FUS = \fourmat{\bI}{\bS}{{\bf 0}}{\bI}\,\,\middle| \,\,\bS\in \Sym(m;2)\right\}.
\end{align}
Let $\cP$ be the semidirect product of $S_D$ and $S_U$, that is,
\begin{equation}\label{e-stabgroup}
    \cP = \{\FDP\FUS\mid \bP\in \GL(m;2),\bS\in \Sym(m;2)\}.
\end{equation}
Note that the order of the multiplication doesn't matter since
\begin{equation}\label{e-DU}
    \FDP\FUS = \bF_U(\bP\bS\bP\T)\FDP,
\end{equation}
and $\bP\bS\bP\T$ is again symmetric.
It is straightforward to verify that $\cP = \cC_0$, and that in general
\begin{equation}\label{e-Bruhat}
    \cC_r = \{\bF_1\bF_{\mathbf{\Omega}}(r)\bF_2\mid \bF_1,\bF_2\in \cP\},
\end{equation}
where 
\begin{equation}
    \bF_{\mathbf{\Omega}}(r) = \fourmat{\Imrr}{\Imr}{\Imr}{\Imrr},
\end{equation}
with $\Imr$ being the block matrix with $\bI_r$ in upper left corner and 0 else and $\Imrr = \bI_m - \Imr$. Note here that $\mathbf{\Omega} = \bF_\mathbf{\Omega}(m)$ and $\mathbf{\Omega}\bF_\mathbf{\Omega}(r)\mathbf{\Omega} = \bF_\mathbf{\Omega}(m-r)$.
Then it follows by~\eqref{e-Bruhat} (and by~\eqref{e-DU}) that every $\bF\in \Sp(2m;2)$ can be written as
\begin{equation}\label{e-Bruhat1}
    \bF = \bF_D(\bP_1)\bF_U(\bS_1)\bF_{\mathbf{\Omega}}(r)\bF_U(\bS_2)\bF_D(\bP_2).
\end{equation}
The above constitutes the Bruhat decomposition of a symplectic matrix; see also~\cite{MR18,PTC20}.
\begin{rem}
It was shown in~\cite{Can18} that a symplectic matrix $\bF$ can be decomposed as
\begin{equation}\label{e-trung}
\bF = \bF_D(\bP_1)\bF_U^{\sf T}(\bS_1)\mathbf{\Omega}\bF_\mathbf{\Omega}(r)\bF_U(\bS_2)\bF_D(\bP_2).
\end{equation}
If we, instead, decompose $\mathbf{\Omega} \bF$ as in~\eqref{e-trung} and insert $\mathbf{\Omega}^2 = \bI_{2m}$ between $\bF_D(\bP_1)$ and $\bF_U^{\sf T}(\bS_1)$, we see that~\eqref{e-trung} is reduced to~\eqref{e-Bruhat1}.
This reduction from a seven-component decomposition to a five-component decomposition is beneficial in quantum circuits design~\cite{MR18,RCKP18}. 
\end{rem}
In what follows we will focus on the right action of $\cP$ on $\Sp(2m;2)$, that is, the right cosets in the quotient group $\Sp(2m;2)/\cP$. It is an immediate consequence of~\eqref{e-Bruhat1} and~\eqref{e-DU} that a coset representative will look like
\begin{equation}\label{e-generic}
    \bF_D(\bP)\bF_U(\bS)\bF_{\mathbf{\Omega}}(r),
\end{equation}
for some rank $r$, invertible $\bP$, and symmetric $\bS$. However, two different invertibles $\bP$ may yield representatives of the same coset. We make this precise below.

\begin{lem}\label{L-RightCoset}
A right coset in $\Sp(2m;2)/\cP$ is uniquely characterized by a rank $r$, an $r\times r$ symmetric matrix $\bS_r\in \Sym(r)$, and a $r$-dimensional subspace $H$ in $\F_2^m$.
\end{lem}

\begin{proof}
Write a coset representative $\bF$ as in~\eqref{e-generic}. This immediately determines $r$. Next, write $\bS$ in a block form
\begin{equation}
    \bS = \fourmat{\bS_r}{\bX}{\bX^{\sf T}}{\bS_{m-r}},
\end{equation}
where $\bS_r,\bS_{m-r}$ are symmetric. Denote $\widetilde{\bS_r},\widehat{\bS}_{m-r}\in \Sym(m;2)$ the matrices that have $\bS_r$ and $\bS_{m-r}$ in upper left and lower right corner respectively and 0 otherwise. Put also
\begin{equation}
    \widetilde{\bX} = \fourmat{\bI_r}{{\bf 0}}{\bX\T}{\bI_{m-r}}.
\end{equation}
With this notation we have
\begin{equation}
    \bF_U(\bS)\bF_{\mathbf{\Omega}}(r) = \bF_U(\widetilde{\bS_r})\bF_{\mathbf{\Omega}}(r)\bF_U(\widehat{\bS}_{m-r})\bF_D(\widetilde{\bX}).
\end{equation}
In other words $\bF_U(\bS)\bF_{\mathbf{\Omega}}(r)$ and $\bF_U(\widetilde{\bS_r})\bF_{\mathbf{\Omega}}(r)$ belong to the same coset. Now consider an invertible
\begin{equation}
    \widetilde{\bP} = \fourmat{\bP_r}{{\bf 0}}{{\bf 0}}{\bP_{m-r}}.
\end{equation}
It is also straightforward to verify that
\begin{align*}
    \bF_U(\widetilde{\bS_r})\bF_{\mathbf{\Omega}}(r)\bF_D(\widetilde{\bP}) & =  \bF_U(\widetilde{\bS_r})\bF_D(\widehat{\bP})\bF_{\mathbf{\Omega}}(r) \\
    & = \bF_D(\widehat{\bP})\bF_U(\widehat{\bP}^{-1}\widetilde{\bS_r}\widehat{\bP}^{- \sf T})\bF_{\mathbf{\Omega}}(r),
\end{align*}
where
\begin{equation}
    \widehat{\bP} = \fourmat{\bP_r^{-\sf T}}{{\bf 0}}{{\bf 0}}{\bP_{m-r}},
\end{equation}
and the second equality follows by~\eqref{e-DU}. Thus $\bF_D(\bP_1)\bF_U(\bS_1)\bF_{\mathbf{\Omega}}(r)$, where 
\begin{equation}\label{e-tran}
    \bP_1: =  \bP\widehat{\bP}, \,\,\bS_1 :=\widehat{\bP}^{-1}\widetilde{\bS_r}\widehat{\bP}^{-\sf T}
\end{equation}
represents the same coset. 
Note that the transformation~\eqref{e-tran} doesn't change the column space $\bC$ (that is, the lower left corner of $\bF$), which is an $r$-dimensional subspace in $\F_2^m$.
\end{proof}

Next, using Schubert cells will choose a canonical coset representative. We will use the same notation as in the above lemma. 
Let $r$ and $\widetilde{\bS_r}$ be as above. To choose $\bP$, think of the $r$-dimensional subspace $H$ from the above lemma as the column space of a matrix $\bH$, which belongs to some Schubert cell $\cC_{\cI}$. We will use the coset representative
\begin{equation}\label{e-can}
    \bF_O(\bP_{\cI}, \bS_r):= \bF_D(\bP_{\cI})\bF_U(\widetilde{\bS_r})\bF_{\mathbf{\Omega}}(r),
\end{equation}
where $\bP_{\cI}$ is as in~\eqref{e-PI}.


Let $\bF\in \Sp(2m;2)$ be in block from as in~\eqref{e-bf}, and assume it is written as 
\begin{equation}\label{e-decomp}
    \bF = \bF_D(\bP^{-\sf T})\bF_U(\widetilde{\bS_r})\bF_\mathbf{\Omega}(r)\bF_D(\bM)\bF_U(\bS).
\end{equation}
Multiplying both sides of~\eqref{e-decomp} on the left with $\bF_D(\bP\T)$ and on the right with $\FUS$, and then comparing respective blocks we obtain
\begin{align}
    \bP\T\bA & = (\widetilde{\bS_r} + \Imrr)\bM, \label{e-b1}\\ 
    \bP\T\bA\bS & = \bP\T\bB + \Imr\bM^{-\sf T},\label{e-b2}\\ 
    \bP^{-1}\bC & = \Imr\bM,\label{e-b3}\\ 
    \bP^{-1}\bC\bS & = \bP^{-1}\bD + \Imrr\bM^{-\sf T}, \label{e-b4}
\end{align}
which we can solve for $\bM, \widetilde{\bS_r}$ and $\bS$, while assuming that we know $\bF$ (and implicitly $\bP$ which can be determined by the column space of the lower-left block of $\bF$). First we find $\bM$. For this, recall that $\widetilde{\bS_r}$ has nonzero entries only on the upper left $r\times r$ block. Thus, it follows by~\eqref{e-b1} that the last $m-r$ rows of $\bM$ coincide with the last $m-r$ rows of $\bP\T\bA$. Similarly, it follows from~\eqref{e-b3} that the first $r$ rows of $\bM$ coincide with the first $r$ rows of $\bP^{-1}\bC$. With $\bM$ in hand we have
\begin{equation}\label{e-b5}
    \widetilde{\bS_r} = \bP\T\bA\bM^{-1} + \Imrr.
\end{equation}
By using~\eqref{e-b3} in~\eqref{e-b4} we see that the first $r$ rows of $\bM\bS$ coincide with first $r$ rows of $\bP^{-1}\bC\bS$. Similarly, by using~\eqref{e-b1} in~\eqref{e-b2}, we see that the last $m-r$ rows of $\bM\bS$ coincide with the last $m-r$ rows of $\bP\T\bA\bS$. Multiplication with $\bM^{-1}$ yields $\bS$.
We collect everything in Algorithm~\ref{algo:alg1}, which gives not only the Bruhat decomposition but also a canonical coset representative.

\begin{algorithm} \caption{Bruhat Decomposition of Symplectic Matrix} \label{algo:alg1}
{\bf Input:} A symplectic matrix $\bF$. 
\begin{algorithmic}
\STATE~1.  \hspace{.02 in} Block decompose $\bF$ to $\bA,\bB,\bC,\bD$ as in~\eqref{e-bf}.
\STATE~2. \hspace{.02 in} $r = \rk(\bC)$.
\STATE~3. \hspace{.02 in} Find $\bP$ as in~\eqref{e-PI} from $\cs(\bC)$. 
\STATE~4. \hspace{.02 in} $\bM_{\rm up}$ is the first $r$ rows of $\bP^{-1}\bC$.
\STATE~5. \hspace{.02 in} $\bM_{\rm lo}$ is the last $m-r$ rows of $\bP\T\bA$.
\STATE~6. \hspace{.02 in} $\bM = \twomatv{\bM_{\rm up}}{\bM_{\rm lo}}$.
\STATE~7. \hspace{.02 in} $\widetilde{\bS_r} = \bP\T\bA\bM^{-1} + \Imrr$.
\STATE~8. \hspace{.02 in} $\bS_r$ is the upper left $r\times r$ block of $\widetilde{\bS_r}$.
\STATE~9. \hspace{.02 in} $\bN_{\rm up}$ is the first $r$ rows of $\bP^{-1}\bD + \Imrr\bM^{-\sf T}$.
\STATE~10. $\bN_{\rm lo}$ is the last $m-r$ rows of $\bP\T\bB - \Imr\bM^{-\sf T}$.
\STATE~11. $\bS = \bM^{-1}\twomatv{\bN_{\rm up}}{\bN_{\rm lo}}$.
\end{algorithmic}
{\bf Output:} $r,\bP,\bS_r,\bM,\bS$
\end{algorithm}

We end this section with a few remarks.

\begin{rem}
One can follow an analogous path by considering left action of $\cP$ on $\Sp(2m;2)$. This follows most directly by the observation that if $\bF = \bF_D(\bP)\bF_U(\bS)\bF_{\mathbf{\Omega}}(r)$ is a right coset representative then $\bF^{-1} = \bF_{\mathbf{\Omega}}(r)\bF_U(\bS)\bF_D(\bP^{-1})$ is a left coset representative.
\end{rem}
\begin{rem}
Note that for the extremal case $r = m$, a coset representative as in~\eqref{e-can} is completely determined by a symmetric matrix $\bS\in \Sym(m;2)$, since in this case, as one would recall, $\bP_{\cI} = \bI_{\cI} = \bI_m$.
\end{rem}
\begin{rem}\label{R-states}
Directly from the definition we have 
\begin{equation}
    |\cP| = |\GL(m;2)|\cdot |\Sym(m;2)| = 2^{m^2}\prod_{i=1}^m(2^i -1),
\end{equation}
which combined with~\eqref{e-sp} yields
\begin{equation}
    |\Sp(2m;2)/\cP| = \prod_{i=1}^m(2^i+1) = |\cL(2m,m)|.
\end{equation}
The above is of course not a coincidence. Indeed, $\Sp(2m;2)$ acts transitively from the right on $\cL(2m,m)$. 
Next, consider $\rs\twomat{{\bf 0}_m}{\bI_m}\in \cL(2m,m)$. If a symplectic matrix $\bF$ as in~\eqref{e-bf} fixes this space, then $\bC = {\bf 0}$ and $\bA$ is invertible.
Additionally, because $\bF$ is symplectic to start with, we obtain $\bD = \bA^{\!-\sf T}$ and $\bA\bB\T =: \bS$ is symmetric. 
Thus $\bB = \bS\bA\!\T$, and $\bF \in \cP$. 
That is, $\cP$ is the stabilizer (in a group action terminology) of $\rs\twomat{{\bf 0}_m}{\bI_m}\in \cL(2m,m)$.
The mapping $\Sp(2m;2)/\cP\longrightarrow \cL(2m,m)$, given by 
\begin{equation}\label{e-bij}
    \bF_O(\bP_{\cI},\bS_r)\longmapsto \rs \left[\!\!\begin{array}{cc}\Imr\bP_{\cI}\T & (\Imr\widetilde{\bS_r}+\Imrr)\bP_{\cI}^{-1}\end{array} \!\!\right]
\end{equation}
is well-defined (because, as mentioned, the upper half of a symplectic matrix is maximal isotropic). It is also injective, and thus bijective due to cardinality reasons. Of course one can have many bijections but we choose this one due to Theorem~\ref{T-ES}.
\end{rem}

\subsection{The Heisenberg-Weyl Group}\label{Sec-HW}
Fix $N = 2^m$, and let $\{\be_0,\be_1\}$ be the standard basis of $\C^2$, which is commonly referred as the \emph{computational basis}. 
For $\bv = (v_1,\ldots,v_m) \in \F_2^m$ set $\be_{\bv}:=\be_{v_1}\otimes\cdots\otimes \be_{v_m}$. 
Then $\{\be_{\bv}\mid \bv\in \F_2^m\}$ is the standard basis of $\C^N \cong (\C^2)^{\otimes m}$. The \emph{Pauli matrices} are
\begin{equation}
\begin{array}{cccc}
\bI_2, & \!\!\sigma_x = \fourmat{0}{1}{1}{0}\mbox, & \!\!\sigma_z = \fourmat{1}{0}{0}{-1}\mbox, & \!\!\sigma_y = i\sigma_x\sigma_z.
\end{array}
\end{equation}
For $\ba,{\bf b}\in \F_2^m$ put 
\begin{equation}
    \bD(\ba,{\bf b}):=\sigma_x^{a_1}\sigma_z^{b_1}\otimes\cdots\otimes \sigma_x^{a_m}\sigma_z^{b_m}.
\end{equation}
Directly by definition we have
\begin{equation}\label{e-XZ}
    \bD(\ba,{\bf 0})\be_{\bv} = \be_{\bv+\ba} \text{ and }\bD({\bf 0},{\bf b})\be_{\bv} = (-1)^{\bf b\scriptsize\T\bv}\be_{\bv},
\end{equation}
and thus, the former is a permutation matrix whereas the latter is a diagonal matrix. Then 
\begin{equation}\label{e-mult}
    \bD(\ba,{\bf b})\bD(\bc,\bd) = (-1)^{\scriptsize{\bf b}\T\bc}\bD(\ba + \bc,{\bf b}+\bd).
\end{equation}
Thanks to~\eqref{e-mult} we have 
\begin{equation}
    \bD(\ba,{\bf b})\bD(\bc,\bd) =  (-1)^{\scriptsize{\bf b}\T\bc+\ba\T\bd}\bD(\bc,\bd)\bD(\ba,{\bf b}).
\end{equation}
In turn, $\bD(\ba,{\bf b})$ and $\bD(\bc,\bd)$ commute iff
\begin{equation}\label{e-symp}
    \inners{\ba,{\bf b}}{\bc,\bd}:= {\bf b}\T\bc+\ba\T\bd = 0,
\end{equation}
that is, iff $(\ba,{\bf b})$ and $(\bc,\bd)$ are orthogonal with respect to the symplectic inner product~\eqref{e-sinner}. Also thanks to~\eqref{e-mult}, the set
\begin{equation}
    \cH\cW_N := \{i^k\bD(\ba,{\bf b})\mid \ba,{\bf b}\in \F_2^m, k = 0,1,2,3\}
\end{equation}
is a subgroup of $\U(N)$ and is called the \emph{Heisenberg-Weyl} group. We will also call its elements Pauli matrices as well. 
Directly from the definition, we have a surjective homomorphism of groups 
\begin{equation}\label{e-PsiN}
\Psi_N : \cH\cW_N \longrightarrow \F_2^{2m}, \,\, i^k\bD(\ba,{\bf b})\longmapsto (\ba,{\bf b}).
\end{equation}
Its kernel is $ \ker \Psi_N = \{\pm \bI_N,\pm i\bI_N\} \cong \Z_4$. 
We will denote $\cH\cW_N^* := \cH\cW_N/\ker \Psi_N$ the \emph{projective} Heisenberg-Weyl group, and the induced isomorphism $\Psi_N^*$. 

Note that $\inners{\sbt}{\sbt}$ defines a nondegenerate bilinear form in $\F_2^{2m}$ that translates commutativity in $\cH\cW_N$ to orthogonality in $\F_2^{2m}$.
A commutative subgroup $\cS\subset \HW$ is called a \emph{stabilizer group} if $-\bI_N\notin \cS$. Thus, for a stabilizer $\cS$, thanks to \eqref{e-symp} we have $\Psi_N(\cS) \subseteq \Psi_N(\cS)^{\sperp}$ \cite{CRSS97, CRSS98}. 
In addition, because $\Psi_N$ restricted to a stabilizer is an isomorphism, we have that $|\cS| = 2^r$ iff $\dim \Psi_N(\cS) = r$. 
We will think of $\Psi_N(\cS)$ as the row space of a full rank matrix $[\bA\,\,\,\,\bB]$ where both $\bA$ and $\bB$ are $r\times m$ binary matrices. We will write 
\begin{equation}
  \bE(\bA,\,{\bf B}) : = \{\bE(\bx\T\bA,\,\bx\T{\bf B})\mid \bx\in \F_2^r\},  
\end{equation}
where $\bE(\ba,\,{\bf b}):=i^{\scriptsize \ba\T{\bf b}}\bD(\ba,{\bf b}).$ Combining this with \eqref{e-XZ} and \eqref{e-mult} we obtain
\begin{equation}\label{e-Eab1}
    \bE(\ba,{\bf b}) = i^{\scriptsize \ba\T{\bf b}}\sum_{\scriptsize \bv\in\F_2^m}(-1)^{\scriptsize{\bf b}\T\bv}\be_{\scriptsize\bv+\ba}\be_{\scriptsize\bv}\!\!\!\T.
\end{equation}
Next, if $\rs [\bA\,\,\,\,\bB]$ is self-orthogonal in $\F_2^{2m}$ then $\bE(\bA,\,\bB)$ is a stabilizer.
Moreover, $\Psi_N(\bE(\bA,\,\bB)) = \rs [\bA\,\,\,\,\bB]$, which yields a one-to-one correspondence between stabilizers in $\HW$ and self-orthogonal subspaces in $\F_2^{2m}$.
It also follows that a maximal stabilizer must have $2^m$ elements. Thus there is a one-to-one correspondence between maximal stabilizers and Lagrangian Grassmannians $\cL(2m,m) \subset \cG(2m,m)$.
Of particular interest are maximal stabilizers
\begin{align}
    \cX_N & := \bE(\bI_m, {\bf 0}_m) = \{\bE(\ba,{\bf 0}) \mid \ba \in \F_2^m\}, \\
    \cZ_N & := \bE({\bf 0}_m,\bI_m) = \{\bE({\bf 0},{\bf b}) \mid {\bf b} \in \F_2^m\},
\end{align}
which we naturally identify with $X_N := \Psi_N(\cX_N) = \rs [\bI_m\,\,\,\,{\bf 0}_m]$ and $Z_N := \Psi_N(\cZ_N) = \rs [{\bf 0}_m \,\,\,\,\bI_m]$.

What follows holds in general for any stabilizer, but for our purposes, we need only focus on the maximal ones. 
Let $\cS = \bE(\bA,\,{\bf B}) \subset \HW$ be a maximal stabilizer and let $\{\bE_1,\ldots,\bE_m\}$ be an independent generating set of $\cS$ (that is, $\spann \{\Psi_N(\bE_1), \ldots,\Psi_N(\bE_m)\} = \Psi_N(\cS)$). Consider the complex vector space \cite{Gottesman97} 
\begin{equation}\label{e-state}
    V(\cS) := \{\bv\in \C^N\mid \bE_i\bv=\bv, i=1,\ldots,m\}.
\end{equation}
It is well-known (see, e.g.,~\cite{NC00}) that $\dim V(\cS) = 2^m/|\cS| = 1$. A unit norm vector that generates it is called \emph{stabilizer state}, and with a slight abuse of notation is also denoted by $V(\cS)$. Because we are disregarding scalars, it is beneficial to think of a stabilizer state as \emph{Grassmannian line}, that is, $V(\cS)\in \cG(\C^N,1)$. Next,
\begin{equation} \label{e-proj}
    \bb{$\Pi$}_{\cS} := \prod_{i=1}^m\frac{\bI_N+\bE_i}{2} = \frac{1}{N}\sum_{\bE\in\cS}\bE
\end{equation}
is a projection onto $V(\cS)$.

Given a stabilizer as above, for any $\bd\in \F_2^m$, $\{(-1)^{d_1}\bE_1,\ldots,(-1)^{d_m}\bE_m\}$ also describes a stabilizer $\cS_{\bd}$. Similarly to \eqref{e-proj} put 
\begin{align}\label{e-projD}
\begin{split}
    \bb{$\Pi$}_{\cS_{\bd}} := & \prod_{i=1}^m\frac{\bI_N+(-1)^{d_i}\bE_i}{2} \\  = & \frac{1}{N}\sum_{{\scriptsize\bx}\in\F_2^m}(-1)^{\bd^{\sf T}\bx}\bE(\bx\T\bA,\,\bx\T{\bf B}).
\end{split}
\end{align}
It is readily verified that $\{\bb{$\Pi$}_{\cS_{\bd}}\mid \bd\in \F_2^m\}$ are pair-wise orthogonal, and a stabilizer group determines a resolution of the identity
\begin{equation}\label{e-projDD}
\bI_N = \sum_{\bd\in\F_2^m} \bb{$\Pi$}_{\cS_{\bd}}.
\end{equation}
Thus every such projection determines a one-dimensional subspace which with another abuse of notation (see also Remark \ref{R-char} below) we call a stabilizer state. 
%

\begin{rem}\label{R-char}
A stabilizer state as in \eqref{e-state} is the unit norm vector that is fixed by the stabilizer. 
Now for every $\bE\in \cS = \bE(\bA,\,{\bf B})$ there exists a unique $\bx\in \F_2^m$ such that $\bE = \bE(\bx\T\bA,\,\bx\T{\bf B})$. 
For $\bb{d}\in \F_2^m$, consider the map $\chi_{\bd}:\bE(\bx\T\bA,\,\bx\T{\bf B})\longmapsto (-1)^{\scriptsize\bd\T\bx}$. Then $\bb{$\Pi$}_{\cS_{\bd}}$ projects onto 
\begin{equation}\label{e-VSD}
    V(\cS_{\bd}):=\{\bv\in \C^N\mid \bE\bv = \chi_{\bd}(\bE)\bv \text{ for all } \bE\in \cS_{\bd}\},
\end{equation}
that is, the state that under the action of $\bE$ is scaled by $\chi_{\bd}(\bE)$.
Then of course $V(\cS_{\scriptsize{\bf 0}}) = V(\cS)$ where ${\bf 0}\in \F_2^m$. 
In addition, the map $\chi_{\bd}$ is a linear character of $\cS$, which has led to non-binary considerations \cite{AK01}.
\end{rem}
\begin{rem}
Let $\{\bE_1,\ldots,\bE_m\}$ be an independent generating set of a maximal stabilizer $\cS$ and consider $\cS_{\bd}$. By~\cite[Prop.~10.4]{NC00} it follows that for each $i=1,\ldots, m$, there exists $\bG_i\in \HW$ such that $\bG_i^\dagger\bE_i\bG_i = -\bE_i$ and $\bG_i^\dagger\bE_j\bG_i = \bE_j$ for $i\neq j$. Now put $\bG_{\bd}: = \bG_1^{d_1}\cdots\bG_m^{d_m}$. Then
\begin{equation}
    \bG_{\bd}^\dagger \bb{$\Pi$}_\cS\bG_{\bd} = \bb{$\Pi$}_{\cS_{\bd}}.
\end{equation}
It follows that $\{V(\cS_{\bd})\mid \bd\in \F_2^m\}$ is an orthonormal basis of $\C^N$. In \cite{TBV17} the authors used a similar insight to construct maximal sets of \emph{mutually unbiased bases}.
\end{rem}

\section{Clifford Group}\label{Sec-Cliff}
The Clifford group in $N$ dimensions \cite{Gottesman09} is defined to be the normalizer of $\HW$ in $\U(N)$ modulo $\U(1)$:
\begin{equation}
    \cl_N=\{\bG\in \U(N)\mid \bG\HW\bG^\dagger = \HW\}/\U(1).
\end{equation}
The reason one quotients out $\U(1) \cong \{\alpha\bI_N\mid |\alpha|=1\}$, is to obtain a finite group. In this case $\HW^*$ is a normal subgroup of $\cl_N$.

Let $\{\be_1,\ldots,\be_{2m}\}$ be the standard basis of $\Fm$, and consider $\bG\in \cl_N$. Let $\bb{c}_i\in \Fm$ be such that 
\begin{equation}
    \bG\bE(\be_i)\bG^\dagger = \pm\bE(\bc_i).
\end{equation}
Then the matrix $\bF_{\scriptsize\bG}$ whose $i$th row is $\bc_i$ is a symplectic matrix such that
\begin{equation}\label{e-cliff}
    \bG\bE(\bc)\bG^\dagger = \pm\bE(\bc\T\bF_{\scriptsize\bG})
\end{equation}
for all $\bc\in \Fm$. Based on~\eqref{e-cliff} we obtain a group homomorphism 
\begin{equation}\label{e-Phi}
    \Phi : \cl_N\longrightarrow \Sp(2m;2),\quad \bG\longmapsto \bF_{\bG},
\end{equation}
with kernel $\ker \Phi = \HW^*$~\cite{RCKP18}. This map is also surjective; see Section~\ref{S-DC} where specific preimages are given. From \eqref{e-sp} and \eqref{e-PsiN} ($|\HW^*| = 2^{2m}$) follows that 
\begin{equation}\label{e-cliff1}
    |\cl_N| = 2^{m^2+2m}\prod_{i=1}^m(4^i-1).
\end{equation}
\begin{rem}\label{R-Gd}
Since $\Phi$ is a homomorphism we have that $\Phi(\bG^\dagger) = \bF_{\bG}^{-1}$ and as a consequence $\bG^\dagger\bE(\bc)\bG = \pm\bE(\bc\T\bF_{\bG}^{-1})$. We will make use of this simple observation later on to determine when a column of $\bG$ is an eigenvector of $\bE(\bc)$. This interplay with symplectic geometry provides an exponential complexity reduction in various applications. Here, we will focus on efficiently computing common eigenspaces of maximal stabilizers.
\end{rem}
The \emph{phase} and \emph{Hadamard} matrices
\begin{equation}\label{e-ph}
    \bG_P=\fourmat{1}{0}{0}{i} \text{ and } \bH_2=\frac{1}{\sqrt{2}}\fourmat{1}{1}{1}{-1}
\end{equation}
are easily seen to be in the Clifford group $\cl_2$.
Some authors also include $(\bG_P\bH_2)^3 = \exp(\pi i/4)\bI_2$~\cite{CRSS98}, which in our case would disappear as a scalar quotient. Thus~\eqref{e-cliff1} differs by a factor of 1/8 of what is commonly considered as the cardinality of the Clifford group; see \href{https://oeis.org/A003956}{A003956} at oeis.org. For our purposes the phases are irrelevant. 

\subsection{Decomposition of the Clifford Group}\label{S-DC}
In this section we will make use of the Bruhat decomposition of $\Sp(2m;2)$ to obtain a decomposition of $\cl_N$. To do so we will use the surjectivity of $\Phi$ from~\eqref{e-Phi} and determine preimages of coset representatives from~\eqref{e-can}. The preimages of symplectic matrices $\FDP,\FUS$, and $\bF_{\mathbf{\Omega}}(r)$ under $\Phi$ are the unitary permutation matrix
\begin{align}\label{e-GDP}
    \bG_D(\bP) & := \be_{\scriptsize\bv}\longmapsto \be_{\scriptsize{\bP\T\bv}},\\\label{e-GDP1}
    \bG_U(\bS) & := \diag\left(i^{\scriptsize{\bv\T\bS\bv} \mod 4}\right)_{\scriptsize\bv\in\F_2^m}, \\
    \bG_{\mathbf{\Omega}}(r) & :=(\bH_2)^{\otimes r}\otimes \bI_{2^{m-r}},\label{e-Gom}
\end{align}
respectively. We refer the reader to~\cite[Appendix I]{RCKP18} for details.
\begin{rem}
Note that directly by the definition of the Hadamard matrix we have
\begin{equation}\label{e-Hadamard}
    \bH_N:=\bG_{\mathbf{\Omega}}(m) = \frac{1}{\sqrt{2^m}}[(-1)^{\scriptsize{\bv\T\bw}}]_{\bv,\bw\in \F_2^m}.
\end{equation}
Whereas, for any $r=1,\ldots,m$, one straightforwardly computes
\begin{equation}\label{e-Gomr}
    \bG_{\mathbf{\Omega}}(r)\!\cdot\!\bZ(m,r) = [(-1)^{\scriptsize{\bv\T\bw}}\cdot f(\bv,\bw,r)]_{\bv,\bw\in \F_2^m},
\end{equation}
where $\bZ(m,r) := \bI_{2^r}\otimes \sigma_z^{\otimes m-r}$ is the diagonal Pauli that acts as $\sigma_z$ on the last $m-r$ qubits, and 
\begin{equation}\label{e-f}
    f(\bv,\bw,r) = \prod_{i=r+1}^m(1+v_i + w_i).
\end{equation}
Note that the value of $f$ will be 1 precisely when $\bv$ and $\bw$ coincide in their last $m-r$ coordinates and 0 otherwise. 
It follows that $f$ is identically 1 when $r = m$ and $f$ is the Kronecker function $\delta_{\bv,\bw}$ when $r=0$. We will use $f$ to determine the \emph{sparsity} or \emph{rank} of a Clifford matrix/stabilizer state. Of course $r = m$ corresponds to \emph{fully occupied} objects with only nonzero entries; see also Remarks~\ref{R-rank0} and~\ref{R-rank1} for the extreme cases of $r = 0, 1$.
\end{rem}

\begin{exa}[\mbox{Example~\ref{OEx} continued}]\label{OEx1} Let us reconsider the invertible matrices from~\eqref{e-ex}.
Recall that there we had $m=3,r=2$. 
Here we will construct the Cliffords corresponding to the canonical coset representative~\eqref{e-bij}, with $\bS_r = \mathbf{0}_{2\times 2}$.
For the case $u=0$ one computes\footnote{See Section~\ref{SS-AS} for why we consider the transpose instead of $\bP$ itself.} $\bG_D(\bP^{\sf T}_{u = 0})$ as in~\eqref{e-GDP1}, and multiplies it (from the right) by $\bG_\mathbf{\Omega}(2)$ as in~\eqref{e-Gom} (we will omit $1/\sqrt{2^2}$) and and then by $\bZ(3,2)$ to obtain
\begin{equation}\label{e-eex2}
  \bG_{u=0} = \left[\!\!\begin{array}{cccccccc}+&0&+&0&+&0&+&0\\+&0&-&0&+&0&-&0\\     0 &   -  &   0   & -   &  0  &  - &    0  &  -\\     0 &   -  &   0  &   +   &  0  &  -   &  0  &   +\\          +   &  0  &   +   &  0  &  -  &   0  &  -   &  0\\     +   &  0  &  -   &  0  &  - &    0  &   +   &  0\\     0   & -  &   0  &  -&     0   &  + &    0   &  + \\      0   & -  &   0   &  +   &  0 &    +  &   0  &  -
\end{array}\!\!\right].
\end{equation}
As mentioned,~\eqref{e-GDP} by definition yields a permutation matrix. Thus $\bG_{u=0}$ is nothing else but $\bG_\mathbf{\Omega}(2) = \bH_2\otimes\bH_2\otimes \bI_2$, with its rows permuted accordingly, and a possible sign introduced to its columns by the diagonal matrix $\bZ(3,2) = \bI_4\otimes \sigma_z$. Similarly, for the case $u = 1$, one obtains
\begin{equation}
  \bG_{u=1} = \left[\!\!\begin{array}{cccccccc}     +  &   0  &   +  &   0 &    + &    0  &   +   &  0\\     +  &   0  &  -&     0  &   +  &   0   & -   &  0 \\      0   & -  &   0   & -   &  0   & -  &   0  &  - \\      0  &  -   &  0   &  +  &   0 &   -  &   0&     + \\
     0 &   -  &   0  &  -   &  0  &   + &    0 &    +\\      0  &  -  &   0  &   +   &  0  &   +  &   0   & -\\
     +  &   0    & +  &   0  &  -   &  0  &  - &    0\\
     +  &   0  &  -  &   0  &  - &    0  &   + &    0
\end{array}\!\!\right].
\end{equation}
We will discuss how the $\{+,-,0\}$ patterns are correlated later on.
\end{exa}
Let us now return to the Clifford group. The Bruhat decomposition~\eqref{e-Bruhat1} of $\Sp(2m;2)$ already gives a decomposition of $\cl_N$. However, in order to have a concise approach one has to be a bit careful. 
In this section we will  write $\bG = \Phi^{-1}(\bF)$, where the equality is taken modulo the center $\HW^*\ltimes \Z_8$. In other words, we will disregard the central part of $\Phi^{-1}(\bF)$ and consider only the Clifford part. The cyclic group $\Z_8$ of order 8 comes into play to accommodate 8th roots of unity coming out of products $\bG_U(\bS)\bG_{\mathbf{\Omega}}(r)$. This setup, yet again, confirms the importance of
\begin{equation}
    \cl^*_N: = \{\exp(i\pi k/4)\bG \mid k\in \Z_8,\bG\in \cl_N\}.
\end{equation}
Let $\cG = \{\bG_D(\bP) \bG_U(\bS)\mid \bP\in \GL(m;2),\bS\in \Sym(m;2)\}$ be the preimage of $\cP$ from~\eqref{e-stabgroup}. For obvious reasons, it is referred as the \emph{Hadamard-free} group; see also~\cite{bravyi2020hadamardfree}. As for the case of the symplectic group, this group acts from the right on matrices of the form
\begin{equation}
     \bG_D(\bP_1) \bG_U(\bS_1)\bG_{\mathbf{\Omega}}(r) \bG_U(\bS_2) \bG_D(\bP_2)
\end{equation}
and thus, a coset representative would look like
\begin{equation}
    \bG_{\scriptsize\bF}:=\bG_D(\bP_1) \bG_U(\bS_1)\bG_{\mathbf{\Omega}}(r),
\end{equation}
One is interested on coset representatives. To understand this, it is enough to understand the preimage of generators of $\GL(m;2)$ and $\Sym(m;2)$. Let us start with the former, which can be generated by two elements \cite{st68}. 
Namely, it can be generated  $\bP:=\bI_m + \bE_{12}$ where $\bE_{12}$ is the elementary (binary) matrix with 1 in position $(1,2)$ and 0 elsewhere and the cyclic permutation matrix $\bb{$\Pi$}_{\text{cycl}}$ acting as the permutation $(12\cdots m)$. It is of interest to consider a larger set of generators. Let $\bb{$\Pi$}_{i,j}$ be a transposition matrix. Then of course $\bP$ along with all the $\bb{$\Pi$}_{i,j}$ generate $\GL(m;2)$. While $\bb{$\Pi$}_{i,j}$ swaps dimensions $i$ and $j$ in $\F_2^m$, it is easily seen that $\Phi^{-1}(\bF_D(\bb{$\Pi$}_{i,j}))$ swaps the tensor dimensions $i$ and $j$ in $(\C^2)^{\otimes m}$. Moreover
\begin{equation}
    \Phi^{-1}(\bF_D(\bP)) = \left[\!\!\begin{array}{cccc}1&0&0&0\\0&1&0&0\\0&0&0&1\\0&0&1&0\end{array}\!\!\right]\otimes\bI_{N-4}.
\end{equation}
The above $4\times 4$ matrix is known in quantum computation as the controlled-NOT (CNOT) quantum gate. In itself, the CNOT gate is of form $\bG_D(\bP)$ where
\begin{equation}
    \bP = \bP^{-1} = \fourmat{1}{1}{0}{1}.
\end{equation}

For $\Sym(m;2)$ we consider matrices $\bS_{\scriptsize\bv}:=\bv\T\bv$ where $\bv\in \F_2^m$ is a vector with at most two non-zero entries. Then
\begin{equation}\label{e-vvt}
    \Phi^{-1}(\bF_U(\bS_\bv)) = \frac{1}{\sqrt{2}}(\bI_N+i\bE({\bf 0},\bv)).
\end{equation}
Note that when $\bv$ has exactly one non-zero entry in position $j$, the $j$th tensor dimension will contain the phase matrix $\bG_P = \exp(-i\pi/4)(\bI_2 + i\sigma_z)$ form~\eqref{e-ph}.
On the other hand, when $\bv$ has exactly two non-zero entries~\eqref{e-vvt} gives rise to $\bG_{CZ}(\bG_P\otimes \bG_P)$ in tensor dimensions $i$ and $j$, where $\bG_{CZ} = \diag(1,1,1,-1)$. 
The latter is known in quantum computation as the controlled-Z (CZ) quantum gate, and it is of form $\bG_U(\bI_2)$. 

In conclusion, the Bruhat decomposition of $\Sp(2m;2)$ directly yields some fundamental quantum gates as described above. Similar ideas were used in \cite{MR18} where the \emph{depth} of \emph{stabilizer circuits} was considered. Classically, there exist several decompositions of the symplectic group, which in principle would yield a decomposition of the Clifford group. Particularly important in quantum computation is the decomposition into \emph{symplectic transvections}~\cite{OMeara}; see~\cite{PRTC20,pllaha2021decomposition} for a detailed description and~\cite{KS14} for using these ideas to sample the Clifford group.

\section{Binary Subspace Chirps}\label{Sec-BSSC}
Binary subspace chirps (BSSCs) were introduced in~\cite{TC19} as a generalization of binary chirps (BCs)~\cite{HCS08}. In this section we describe the geometric and algebraic features of BSSCs, and use their structure to develop a reconstruction algorithm. For each $1\leq r\leq m$, subspace $H\in \cG(m,r;2)$, and symmetric $\bS_r \in \Sym(r)$ we will define a unit norm vector in $\C^N$ as follows. Let $\HI\in \cC_\cI$ be such that $\cs(\bH_{\cI}) = H$, as described in Section~\ref{S-SC}. Then $\HI$ is completed to an invertible $\bP:=\PI$ as in~\eqref{e-PI}. For all ${\bf b}, \ba \in \F_2^m$ define
\begin{equation}\label{e-bssc}
    \bw_{\scriptsize{\bf b}}^{\scriptsize{H,\bS_r}}(\ba) = \frac{1}{\sqrt{2^r}}i^{\scriptsize{\ba\T\bP^-\T\bS\bP^{-1}\ba+2{\bf b}\T\bP^{-1}\ba}}f({\bf b},\bP^{-1}\ba,r),
\end{equation}
where $\bS\in \Sym(m;2)$ is the matrix with $\bS_r$ on the upper left corner and 0 elsewhere, $f$ is as in~\eqref{e-f}, and the arithmetic in the exponent is done modulo 4. 
To avoid heavy notation however we will omit the upper scripts. Then we define a \emph{binary subspace chirp} to be 
\begin{equation}
\bw_{\scriptsize{\bf b}} := \big[\bw_{\scriptsize{\bf b}}(\ba)\big]_{\scriptsize\ba\in \F_2^m}\in \C^N.
\end{equation}
Note that when $r=m$ we have $\bP = \bI_m$ and $f$ is the identically 1 function. Thus, we obtain the \emph{binary chirps} \cite{HCS08}.

Directly from the definition (and the definition of $f$) it follows that $\bw_{\scriptsize{\bf b}}(\ba)\neq 0$ precisely when ${\bf b}$ and $\bP^{-1}\ba$ coincide in their last $m-r$ coordinates. 
Making use of the structure of $\bP$ as in \eqref{e-PI} we may conclude that $\bw_{\scriptsize{\bf b}}(\ba) \neq 0$ iff
\begin{equation}\label{e-bmr}
    \wHI^{\sf T}\ba = {\bf b}_{m-r},
\end{equation}
where ${\bf b}_{m-r}\in \F_2^{m-r}$ consists of the last $m-r$ coordinates of ${\bf b}$. 
Note that since $\rk\HI = r$ it follows that~\eqref{e-bmr} has $2^r$ solutions, and this in turn implies that $\bw_{\scriptsize{\bf b}}$ has $2^r$ non-zero entries. In particular $\wb$ is a unit norm vector. 
Concretely, making use~\eqref{e-PI1} we see that the solution space of~\eqref{e-bmr} is given by 
\begin{equation}\label{e-bmr1}
    \{\widetilde{\bx}:=\bI_{\widetilde{\cI}}{\bf b}_{m-r} + \HI\bx\mid \bx\in \F_2^r\}.
\end{equation}
We say that the rank $r$ determines the \emph{sparsity} of $\wb$ and the subspace $H$ (equivalently $\HI$) determines the \emph{on-off pattern} of $\wb$.
\begin{rem}\label{R-emb}
Fix a subspace chirp $\wb$, and write ${\bf b}\T = [{\bf b}_r^{\sf T}\,\,\,\, {\bf b}^{\sf T}_{m-r}]$. Then $\wb(\ba) \neq 0$ iff $\ba$ is as in \eqref{e-bmr1} for some $\bx\in \F_2^r$. Making use of \eqref{e-PIT} and \eqref{e-PI1} we obtain
\begin{equation}\label{e-p1a1}
    \bP^{-1}\ba = \left[ \!\!\!\begin{array}{c}\bx \\ {\bf b}_{m-r}\end{array}\!\!\!\right],
\end{equation}
and as a consequence $\ba\T\bP^-\T\bS\bP^{-1}\ba = \bx\T\bS_r\bx$ where $\bS_r$ is the (symmetric) upper-left $r\times r$ block of $\bS$.
Thus the nonzero entries of $\wb$ are of the form
\begin{equation}\label{e-p1a}
    \wb(\bx) = \frac{(-1)^{\wt(\scriptsize{{\bf b}_{m-r})}}}{\sqrt{2^r}}i^{\scriptsize{\bx\T\bS_r\bx} + 2{\bf b}\T{\!\!\!_r}\bx}
\end{equation}
for $\bx\in\F_2^r$. Note that there is a slight abuse of notation where we have identified $\bx$ with $\bP^{-1}\ba$ (thanks to \eqref{e-p1a1} and the fact that ${\bf b}$ is fixed). 
Above, the function $\wt(\sbt)$ is just the Hamming weight which counts the number of non-zero entries in a binary vector. 
We conclude that the \emph{on-pattern} of a rank $r$ binary subspace chirp is just a binary chirp in $2^r$ dimensions; compare \eqref{e-p1a} with~\cite[Eq.~(5)]{HCS08}. 
It follows that all lower-rank chirps are embedded in $2^m$ dimensions, which along with all the chirps in $2^m$ dimensions yield all the binary subspace chirps. As discussed, the embeddings are determined by subspaces.
\end{rem}

\subsection{Algebraic Structure of BSSCs}\label{SS-AS}

In what follows we fix a rank $r$, invertible $\bP$, and symmetric $\bS$. Recall that $\bS$ contains a $r\times r$ symmetric in its upper left corner and 0 otherwise. 
Next, let $\bF:=\bF_{\mathbf{\Omega}}(r)\FUS\bF_D(\bP\T)$ and let $\bG_{\bF} = \bG_D(\bP\T) \bG_U(\bS)\bG_{\mathbf{\Omega}}(r)$. 
With this notation we have $\Phi(\bG_\bF) = \bF$. Recall also that $\{\be_{\scriptsize\ba}\mid \ba\in \F_2^m\}$ is the standard basis of $\C^N$.
With a substitution $\bu:= \bP^{-1}\ba$ we have 
\begin{align}
    \bw_{\scriptsize{\bf b}} & = \frac{1}{\sqrt{2^r}}\sum_{\scriptsize\ba\in \F_2^m} \bw_{\scriptsize{\bf b}}(\ba)\be_{\scriptsize\ba} \nonumber \\
    & = \frac{1}{\sqrt{2^r}}\sum_{\scriptsize\bu\in \F_2^m}i^{\scriptsize{\bu\T\bS\bu}}(-1)^{\scriptsize{{\bf b}\T\bu}}f({\bf b},\bu,r)\be_{\scriptsize{\bb{Pu}}} \nonumber\\
     & = \bG_D(\bP\T)\cdot\frac{1}{\sqrt{2^r}} \sum_{\scriptsize\bu\in \F_2^m}i^{\scriptsize{\bu\T\bS\bu}}(-1)^{\scriptsize{{\bf b}\T\bu}}f({\bf b},\bu,r)\be_{\scriptsize{\bu}}\nonumber \\
    & = \bG_D(\bP\T)\bG_U(\bS)\bG_{\mathbf{\Omega}}(r)\bZ(m,r)\be_{\scriptsize{{\bf b}}} \label{e-eqq1}\\
    & = \bG_{\scriptsize{\bF}}\cdot\bZ(m,r)\be_{\scriptsize{{\bf b}}}\label{e-eqq2},
\end{align}
where~\eqref{e-eqq1} follows by~\eqref{e-Gomr}. Note that in~\eqref{e-eqq2}, the diagonal Pauli $\bZ(m,r)$ only ever introduces an additional sign on columns of $\bG_{\scriptsize{\bF}}$. Thus, the binary subspace chirp $\bw_{\scriptsize{{\bf b}}}$ is nothing else but the ${\bf b}$th column of $\bG_{\scriptsize{\bF}}$, up to a sign. However, as mentioned, for our practical purposes a sign (or even a complex unit) is irrelevant.

\begin{exa}[\mbox{Examples~\ref{OEx} and~\ref{OEx1} continued}] Let us consider the case $u = 0$, and for simplicity, let us set the symmetric $\bS$ to be the zero matrix\footnote{$\bG_U(\bS)$ does not affect the on-off pattern at all.}, so that $\bG_U(\bS)$ is the identity matrix. 
The on-off pattern of the resulting BSSCs is governed by the $r=2$ dimensional subspace $H = \{000\T,100\T,001\T,101\T\} = \cs(\bH_{\{1,3\}})$. 
The above argument tells us that these BSSCs are precisely the columns of $\bG_{u=0}$ from~\eqref{e-eex2}. 
One verifies this directly using the definition~\eqref{e-bssc}.
Moreover, the structure of the on-off patterns is completely determined by~\eqref{e-bmr1}. Indeed, we see in~\eqref{e-eex2} two on-off patterns: one determined by $H$ (if $\bI_{\widetilde{\cI}}{\bf b}_{m-r}\in H$)  and one determined by its coset\footnote{There are exactly $2 = 2^3/2^2$ cosets since $H$ has dimensions 2.} (if $\bI_{\widetilde{\cI}}{\bf b}_{m-r}\notin H$). 
In our specific case, we have $\cI = \{1,3\}$, and thus $\bI_{\widetilde{\cI}} = \bI_{\{2\}} = 010\T$. Thus $\bI_{\widetilde{\cI}}{\bf b}_{m-r}\in H$ iff ${\bf b}_{m-r} = 0$ iff ${\bf b} \in \{000\T,010\T,100\T,110\T\}$, which corresponds to columns $\{1,3,5,7\}$. 
Additionally, within each of these columns, the on-off pattern is again governed by $H$. 
Indeed, the non-zero entries in these columns are in positions/rows indexed by $H$, that is, $\{1,2,5,6\}$ -- precisely as described by~\eqref{e-bmr}. 
Since cosets form a partition it follows immediately that columns indexed by different cosets are orthogonal. Orthogonality of columns within each coset is a bit more delicate to see directly. 
We will further discuss the general structure of on-off patterns in Section~\ref{Sec-OO}.

\end{exa}

Equation~\eqref{e-bij} gives a one-to-one correspondence between canonical coset representatives and maximal stabilizers. Above we mentioned that BSSCs are columns of Clifford matrices parametrized by such coset representatives. The last piece of the puzzle is found by simultaneously diagonalizing the commuting matrices of a maximal stabilizer. We make this precise in the following.
\begin{theo}\label{T-ES}
Let $\bF$ and $\bG_{\scriptsize\bF}$ be as above. The set $\{\bw_{\scriptsize{\bf b}}~\mid{\bf b}\in \F_2^m\}$ consisting of the columns of $\bG_{\scriptsize\bF}$ is the common eigenspace of the maximal stabilizer $\bE(\Imr\bP\T,\,(\Imr\bS+\Imrr)\bP^{-1})$ from~\eqref{e-bij}.
\end{theo}
\begin{proof}
Consider the matrix $\bG:=\bG_{\scriptsize{\bF}}$ parametrized by the symplectic matrix $\bF$, and recall that $\wb$ is the ${\bf b}$th column of $\gf$. It follows from Remark \ref{R-Gd} that the columns of $\bG$ are the eigenspace of $\bE(\bx,\by)$ iff 
\begin{equation}
    \bG^\dagger \bE(\bx,\by)\bG = \pm \bE([\bx,\by]\T\bF^{-1})
\end{equation}
is diagonal.
%
%
%
%
Recall also that $\bE(\bx,\by)$ is diagonal iff $\bx = {\bf 0}$, and observe that $\bF_{\mathbf{\Omega}}(r)^{-1} = \bF_{\mathbf{\Omega}}(r)$. 
Thus, $\bG_{\mathbf{\Omega}}(r)$ will be the common eigenspace of the maximal stabilizer $\cS$ iff $\pm\bE([\bx\,\,\by]^{\T}\bF_{\mathbf{\Omega}}(r))$ is diagonal for all $\bE(\bx,\by)\in \cS$. 
Then it is easy to see that such maximal stabilizer is $\bE(\Imr,\,\Imrr)$. Next, if $\bw$ is an eigenvector of $\bE(\bc)$ then
\begin{align*}
   \bG\bw = \pm\bG\bE(\bc)\bw & = \pm\bG\bE(\bc)\bG^\dagger\bG\bw = \pm\bE(\bc\T\Phi(\bG))\bG\bw
   \end{align*}
implies that $\bG\bw$ is an eigenvector of $\bE(\bc\T\Phi(\bG))$. The proof is concluded by computing $[\Imr\,\,\,\,\Imrr]\bF_U(\bS)\bF_D(\bP\T)$.
\end{proof}
\begin{rem}
Note that for $r=m$ one has $\bE(\Imr,\,\Imrr) = \bE(\bI_m,\, {\bf 0}_m)$ and $\bG_{\mathbf{\Omega}}(r) = \bH_N$. Thus the above theorem covers the well-known fact that $\bH_N$ is the common eigenspace of the maximal stabilizer $\cX_N = \bE(\bI_m,\, {\bf 0}_m)$. 
It is also well-known that the standard basis of $\C^N$ (that is, $\bI_N$) is the common eigenspace of the maximal stabilizer $\cZ_N= \bE({\bf 0}_m,\bI_m)$ of diagonal Paulis. This is of course consistent with the aforesaid fact since $[{\bf 0}_m\,\,\,\,\bI_m]\mathbf{\Omega} = [\bI_m\,\,\,\,{\bf 0}_m]$ and $\bH_N = \Phi^{-1}(\mathbf{\Omega})$.
In this extremal case we also have $\bP_{\cI} = \bI_m$ and $\widetilde{\bS_r} = \bS\in \Sym(m;2)$. So the above theorem also covers~\cite[Lem.~11]{CRCP19} which (in the language of this paper) says  that the common eigenspace of $\bE(\bI_m,\,\bS)$ is $\bG_U(\bS)\bH_N$.
\end{rem}
\begin{rem}
Theorem~\ref{T-ES} is a closed form realization of a more general fact. 
Let $\cS$ be a maximal stabilizer and let $S = \rs [\bA\,\,\,\, \bB] \subset \Fm$ be its corresponding isotropic subspace. 
Consider also the diagonal Paulis $\cZ_N$ and its corresponding subspace $Z_N = \rs [{\bf 0}_m\,\,\,\,\bI_m]$. 
Then, by~\cite[Alg.~1]{RCKP18} there exists $\bG\in \cl_N$ such that $\bG\cS\bG^\dagger = \cZ_N$. 
In other words, $\bG^\dagger$ simultaneously diagonalizes $\cS$, and moreover, the respective diagonal is a Pauli. 
In the symplectic domain, it follows by~\cite[Thm.~25]{RCKP18} that there are precisely $2^{m(m+1)/2}$ symplectic solutions to the equation $[\bA\,\,\,\, \bB]\bF = [{\bf 0}_m\,\,\,\,\bI_m]$.
\end{rem}

\begin{cor}\label{C-size}
The set of all binary subspace chirps $\bssc$ consists of $2^m\prod_{r=1}^m(2^r+1)$
unit norm vectors, while the set of all binary chirps $\chirp$ consists of $2^{m(m+3)/2}$ unit norm vectors. The ratio of codebook cardinalities is $|\bssc|/|\chirp|\approx 2.384$.
\end{cor}
\begin{proof}
By Lemma~\ref{L-RightCoset} we need only consider coset representatives~\eqref{e-can}, which are in bijection~\eqref{e-bij} with the set of maximal stabilizers $\cL(2m,m)$. Now the first statement follows by~\eqref{e-Lan}.

Next, recall that the binary chirps correspond to those binary subspace chirps for which $r=m$, and thus, $\bP = \bI_m$. In other words they are parametrized by a symmetric matrix $\bS\in\Sym(m;2)$ and ${\bf b}\in \F_2^m$.
\end{proof}

The size of the codebook $\bssc$ can be also straightforwardly deduced from the characterization of Lemma~\ref{L-RightCoset}. Indeed, each rank zero BSSC has precisely $2^0 = 1$ non-zero entry. Thus there are $2^m$ rank zero BSSCs, parametrized only by a column index ${\bf b}\in\F_2^m$. On the other hand, rank $r$ BSSCs are characterized by $\bS\in \Sym(r;2)$ and $H\in \cG(m,r;2)$. Of course, $|\Sym(r;2)| = 2^{r(r+1)/2}$. Whereas the size of the Grassmannian is given by the 2-binomial coefficient, that is
\begin{equation}
    |\cG(m,r;2)| = {m\choose r}_2 = \prod_{i=0}^{r-1}\frac{1-2^{m-i}}{1-2^{i+1}}.
\end{equation}
Thus, we have
\begin{align}\label{e-BSSCsSumProd}
    |\bssc| = 2^m\cdot \sum_{r= 0}^m2^{r(r+1)/2}{m\choose r}_2 = 2^m\cdot\prod_{r=1}^m(2^r+1),
\end{align}
where the last equality is simply the 2-binomial theorem~\cite{andrews}.
\begin{cor}\label{C-states}
Each binary subspace chirp is a stabilizer state. The converse is also true.
\end{cor}
\begin{proof}
Stabilizer states can be defined equivalently as the orbit of $\be_{\scriptsize{\bf 0}}$ under the action of $\cl_N$; see~\cite{AG04} for instance. Then the first statement follows by~\eqref{e-eqq2}. The converse is true due to cardinalities.
\end{proof}
\begin{cor}\label{C-SZN}
Let $\cS$ be a maximal stabilizer. Then the stabilizer state $V(\cS)$ is a rank $r$ BSSC iff $|\cS\cap \cZ_N| = 2^{m-r}$.
\end{cor}
\begin{proof}
By Corollary~\ref{C-states} we know that $V(\cS)$ is a BSSC of rank $r$, which in turn is stabilized by the maximal stabilizer of Theorem~\ref{T-ES}. Such stabilizer has precisely $2^{m-r}$ diagonal Paulis; see also~\eqref{e-ES1}.
\end{proof}

We mentioned that the extremal case $r = m$ gives the codebook $\chirp$. Before discussing general on-off patterns, we consider the lower-end extremal cases $r = 0,1$.

\begin{rem}\label{R-rank0}
Let $r = 0$. In this case we again have $\bP = \bI_m$ and $\bS = {\bf 0}_m$. 
In addition $f(\bv,\bw,0) = \delta_{\bv,\bw}$.
Thus, from~\eqref{e-bssc} we see that $\wb(\ba) \neq 0$ iff $\ba = {\bf b}$, in which case we have $\wb(\ba) = (-1)^{\wt({\bf b})}$. 
This can also be seen from~\eqref{e-eqq2}. Indeed, since $\bG_\mathbf{\Omega}(0) = \bI_N$ we have $\bG_\bF = \bI_N$. Note also that $\bZ(m,0) = \sigma_z\otimes\cdots\otimes\sigma_z = \bE({\bf 0},{\bf 1})$ is the common eigenspace of the maximal stabilizer $\bE(\bI_m,\bI_m)$, as established by Theorem~\ref{T-ES}.
\end{rem}
\begin{rem}\label{R-rank1}
Let $r = 1$. In this case, either $\bS = {\bf 0}_m$ or $\bS = \be_1\be_1\T$, where $\be_1\in \F_2^m$ is the first standard basis vector. 
It follows that, up to a Pauli matrix, $\bG_U(\bS)$ is either $\bI_N$ or the \emph{transvection} $(\bI_N + iZ_1)/\sqrt{2}$ where $Z_1 = \bE({\bf 0}, \be_1)$ has $\sigma_z$ on the first qubit and identity elsewhere; see also~\eqref{e-vvt}. Similarly $\bG_\mathbf{\Omega}(1) = (X_1 + Z_1)/\sqrt{2}$ is another transvection. 
Thus, rank one BSSCs are columns of transvections, permuted by some Clifford permutation $\bG_D(\bP)$. See~\cite{PRTC20,pllaha2021decomposition} for more on transvections transvections.
\end{rem}

\begin{figure}
\centering
\subfloat[Subfigure 1 list of figures text][$r=0$.]{
\includegraphics[width=0.1\columnwidth]{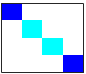}
\label{fig:subfig1}}
\hspace*{0.22 in}
\subfloat[rank1][$r=1$.
]{
\includegraphics[width=0.5\columnwidth]{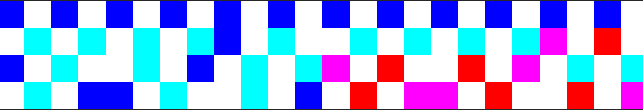}
\label{fig:subfig2}}
\qquad
\subfloat[Subfigure 3 list of figures text][$r=2$ (BCs).
]{
\includegraphics[width=0.66\columnwidth]{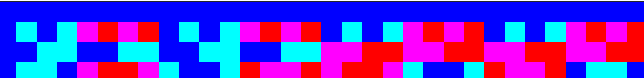}
\label{fig:subfig3}}
\caption{BSSCs in $N=4$ dimensions. White = 0, Blue = 1, Cyan = -1, Red = $i$, Magenta = $-i$.}
\end{figure}

\begin{exa}
Let $m=2$. There are $3 = {2\choose 1}_2$ one dimensional spaces in
$\F_2^m$ and there are two $1\times 1$ symmetric matrices. Thus there
are $2^2\cdot 3\cdot 2 = 24$ BSSCs of rank $r = 1$ in $N = 2^m = 4$
dimensions, as depicted in Figure~\ref{fig:subfig2}. Furthermore, there are eight $2\times 2$ symmetric matrices, and these yield $32 = 2^2\cdot 8$ BCs, as depicted in Figure~\ref{fig:subfig3}. 
Along with $4 = 2^2$ BSSCs of rank 0 depicted on Figure~\ref{fig:subfig1}, we have in total $60 = 4 + 24 + 32 = 4\cdot 3\cdot 5$ BSSCs in $N = 4$ dimensions, as given by~\eqref{e-BSSCsSumProd}.
\end{exa}

\subsection{Structure of On-Off Patterns}\label{Sec-OO}
As discussed, for $\bS_r\in \Sym(r)$ and $H \in \cG(m,r;2)$ we obtain a unitary matrix 
\begin{equation}
    \bU_{H,\scriptsize\bS_r}(\ba,{\bf b}) = \big[\wb(\ba)\big]_{\scriptsize{\ba,{\bf b}}}\in \U(N).
\end{equation}
We will omit the subscripts when the context is clear. We know from~\eqref{e-eqq2} that such a matrix is, up to a diagonal Pauli, an element of $\cl_N$. 
The subspace $H$ determines the \emph{sparsity} of $\bU$. Indeed, we see from \eqref{e-bmr1} that the on-off pattern is supported either by $H$ or by a coset of it. 
Thus, the on-off patterns of different columns are either equal or disjoint. It also follows that in $\bU$ there are $2^{m-r}$ different on-off patterns, each of which repeat $2^r$ times. 

In~\cite{CHJ10} it was shown that BCs form a group under coordinate-wise multiplication. Whereas, we can immediately see that this is not the case for BSSCs. 
For instance, if one considers two BSSCs with disjoint on-off patterns then they coordinate-wise multiply to ${\bf 0}\in\C^N$. When two BSSCs have the same on-off pattern the coordinate-wise multiplication can be determined as follows. 
Let $\bw_1$ and $\bw_2$ be two columns of $\bU$ with the same on-off pattern, indexed by ${\bf b}_1$ and ${\bf b}_2$ respectively. 
Let $\rr = m-r$. In such case, again by~\eqref{e-bmr1}, we must have ${\bf b}_{1,\rr} = {\bf b}_{2,\rr}$, that is, they are equal in their last $\rr$ coordinates. 
Recall also that the non-zero coordinates of a BSSC are determined by \eqref{e-p1a1}. We have that
\begin{equation}
    2^r\bw_1(\ba)\bw_2(\ba) = (-1)^{\scriptsize{\bx\T}\bS\bx + ({\bf b}_{1,r}+{\bf b}_{2,r})\T\bx},
\end{equation}
where $\bx\in \F_2^r$ is such that 
\begin{equation}
        \bP^{-1}\ba = \left[ \!\!\!\begin{array}{c}\bx \\ {\bf b}_{1,\rr}\end{array}\!\!\!\right] = \left[ \!\!\!\begin{array}{c}\bx \\ {\bf b}_{2,\rr}\end{array}\!\!\!\right].
\end{equation}
The matrix $\bP$ above corresponds to $H$ as usual.
Next, the map $\bx\longmapsto \bx\T\bS\bx$ is additive modulo 2, and thus it is of form $\bx\longmapsto {\bf b}_{\tiny{\bS}}^{\!\sf T}\bx$ for some ${\bf b}_{\tiny{\bS}}\in \F_2^m$.
It follows that 
\begin{equation}\label{e-OnOff}
    2^r\bw_1(\ba)\bw_2(\ba) = (-1)^{\scriptsize{ ({\bf b}_{\tiny\bS}+{\bf b}_{1,r}+{\bf b}_{2,r})\T\bx}}.
\end{equation}
Then it is easy to see that the right-hand-side of \eqref{e-OnOff} is, up to a sign, a column of $\bG_D(\bP\T)\bG_\mathbf{\Omega}(r)$.

With a similar argument, when two BSSCs with the same on-off pattern, but different symmetric matrices $\bS_1$ and $\bS_2$, are coordinate-wise multiplied, we obtain, up to sign, a column of $\bG_D(\bP\T)\bG_U(\bS_1+\bS_2)\bG_\mathbf{\Omega}(r)$.
In all cases, the ``up to sign'' is determined by $\wt({\bf b}_\rr)$, that is, the Hamming weight of the last $\rr$ coordinates of the column index. 
The latter is in turn precisely captured by $\bZ(m,r)$; see also \eqref{e-p1a} and \eqref{e-eqq2}.

Also with a similar argument, one determines the conjugate of BSSCs and the coordinate-wise multiplication of BSSCs with $H_1\in\cG(m,r_1)$ and $H_2\in \cG(m,r_2)$. Without diving in details, in this case the on-off pattern will be determined by $H_1\cap H_2$ and of course the sparsity will be $r = \dim H_1 \cap H_2$.

In particular, we have proved the following.

\begin{theo}\label{T-SemiGroup}
The set $\bssc$ is closed with respect to coordinate-wise conjugation. The set $\bssc \cup \{{\bf 0}_N\}$ is closed with respect to coordinate-wise multiplication. The set of all BSSCs of given sparsity $r$ and on-off pattern is isomorphic to $\Sym(r;2)$.
\end{theo}

\subsection{BSSCs as Grassmannian Line Codebooks}
As discussed, the codebooks $\bssc$ and $\chirp$ are codebooks in $\cG(N,1)$. The cardinalites were determined in Corollary~\ref{C-size}. In order to have a complete comparison, one needs to also consider the relevant metric, which for Grassmannian lines codebooks is the \emph{chordal distance}
\begin{equation}\label{e-chordal}
    \chor(\bw_1,\bw_2) = \sqrt{1 - |\bw_1^\dagger\bw_2|^2}.
\end{equation}
Then the \emph{minimum distance} of a codebook is the minimum over all different codewords.

Fix a BC $\bw_1\in\chirp$ parametrized by $\bS_1 \in \Sym(m)$, and let $\bw_2$ range among $2^m$ BCs parametrized by $\bS_2\in \Sym(m)$. Then~\cite{HCS08,CHJ10,CRCP19}
\begin{equation}\label{e-BCinner}
    |\bw_1^\dagger\bw_2|^2 = \begin{cases} 1/2^r, & 2^r\,\,\text{times},\\ 0, & 2^m-2^r\,\,\text{times}, \end{cases}
\end{equation}
where $r = \rk(\bS_1 - \bS_2)$.
It follows immediately that $|\bw_1^\dagger\bw_2| \leq 1/\sqrt{2}$, and thus the minimum distance of the codebook $\chirp$ is $1/\sqrt{2}$. 

For BSSCs we have the following.
\begin{prop}
The codebook $\bssc$ has minimum distance $1/\sqrt{2}$.
\end{prop}
\begin{proof}
As before, it is sufficient to show that $|\bw_1^\dagger\bw_2| \leq 1/\sqrt{2}$ for all $\bw_1,\bw_2\in\bssc$.
By Theorem~\ref{T-SemiGroup}, $\bw_1^\dagger$ is again (the transpose of) a BSSC. 
Then, the inner product $\bw_1^\dagger\bw_2$ is related to the coordinate-wise multiplication of $\bw_1^\dagger$ and $\bw_2$, which as we saw, is either the zero vector or some other BSSC. 
In addition, we know from Remark~\ref{R-emb} that the non-zero entries of BSSCs are lower dimensional BCs. Now the result follows.
\end{proof}
%

The codebooks $\chirp$ and $\bssc$ have the same minimum distance, and by Corollary~\ref{C-size} the latter is 2.384 bigger. Thus, from a coding prospective the codebook $\bssc$ provides a clear improvement. Additionally, we will see next that $\bssc$ can be decoded with similar complexity as $\chirp$. For these reasons, $\bssc$ is an optimal candidate for extending $\chirp$ also from a communication prospective. 
The alphabet of $\chirp$ is $\{\pm 1,\pm i\}$ whereas the alphabet of $\bssc$ is $\{\pm 1,\pm i\}\cup \{0\}$, which is a minimal extension from the implementation complexity prospective.

\begin{cor}\label{C-BCsBSSCs}
Let $\bG_j = \bG_U(\bS_j)\bH_N \in \cl_N$ for $j = 1,2$ and $\bS_j\in \Sym(m;2)$. Then $\bG = \bG_1^\dagger\bG_2$ has sparsity $r$ where $r = \rk(\bS_1+\bS_2)$ and its on-off pattern is determined by $H = \rs(\bS_1+\bS_2)$.
\end{cor}
\begin{proof}
Recall that $\bG_j$ constitutes all the BCs parametrized by $\bS_j$. 
Then the statement follows directly by~\eqref{e-BCinner}.
\end{proof}
\begin{rem}
The vector space of symmetric matrices can be written in terms of chain of nested subspaces, referred in literature as Delsarte-Goethals sets,
\begin{equation}
    \DG(m,0)\subset\DG(m,1) \subset \cdots \subset \DG(m,(m-1)/2)
\end{equation}
with the property that every nonzero matrix in $\DG(m,r)$ has rank at least $m-2r$~\cite{DG75,HKCSS94}. 
For application in deterministic compressed sensing, random access, and quantum computation see~\cite{CJ-LASSO,HCS08,CRCP19}. Since $\DG(m,r)$ is a vector space, it comes with the property that the sum of every two different matrices also has rank at least $m-2r$. 
Thus, for $\bS_1,\bS_2\in \DG(m,(m-r)/2)$, the construction of Corollary~\ref{C-BCsBSSCs} yields a Clifford matrix of sparsity at least $r$. This is an alternative way of creating rank $r$ BSSCs in terms of BCs. However, this will not yield all the BSSCs because not every subspace $H$ is the row/column space of a symmetric matrix $\bS$.
\end{rem}


\section{Reconstruction Algorithms}
\label{Sec-Recon}

In this section we use the rich algebraic structure of BSSCs to construct a low complexity reconstruction/decoding algorithm.
We will build our way up by starting with the reconstruction of a single BSSC. In order to gain some intuition we disregard noise at first.
The problem in hand is to recover $H,\bS_r$, and ${\bf b}$ given a binary subspace chirp $\wb$ as in \eqref{e-bssc}.
In this noiseless scenario, the easiest task is the recovery of the rank $r$. Namely, by \eqref{e-bmr1} we have
\begin{equation}
    \wb(\ba)\ov{\wb(\ba)} = \left\{\!\!\begin{array}{ll} 1 /2^r, & 2^r \text{ times,} \\ 0, & 2^{m-r}\text{ times.}\end{array} \right.
\end{equation}
To reconstruct $\bS_r$ and then eventually $H$ we generalize the \emph{shift and multiply} technique used in~\cite{HCS08} for the reconstruction of binary chirps. 
The underlying structure that enables this generalization is the fact that the on-pattern of BSSC is a BC of lower rank as discussed in Remark~\ref{R-emb}.
However, in our scenario extra care is required as the shifting can perturb the on-off pattern. Namely, we must use only shifts $\ba\longmapsto \ba+\be$ that preserve the on-off pattern. 
It follows by \eqref{e-bmr} that we must use only shifts by $\be$ that satisfy $(\widetilde{\bH_\cI})\T\be = {\bf 0}$, or equivalently $\be = \bH_\cI\by$ for $\by\in \F_2^r$. In this instance, thanks to \eqref{e-PI1} we have 
\begin{equation}
    \bP^{-1}\be = \bP^{-1}\bH_\cI\by = \twomatv{\by}{{\bf 0}}.
\end{equation}
If we focus on the nonzero entries of $\wb$ and on shifts that preserve the on-off pattern of $\wb$ we can make use of Remark \ref{R-emb}, where with another slight abuse of notation we identify $\by$ with $\bP^{-1}\be$. 
It is beneficial to take $\by$ to be ${\bf f}_i$ - one of the standard basis vectors of $\F_2^r$. 
With this preparation we are able to use the shift and multiply technique, that is, shift the given BSSC $\wb$ according to the shift $\bx\longmapsto \bx+{\bf f}_i$ (which only affects the on-pattern and fixes the off-pattern) and then multiply by its conjugate:
\begin{equation}\label{e-shift}
    \wb(\bx+{\bf f}_i)\ov{\wb(\bx)} = \frac{1}{2^r}\cdot i^{\scriptsize {\bf f}\T_{\!\!i}\bS_r{\bf f}_i}\cdot(-1)^{\scriptsize {\bf b}\T_{\!\!r}{\bf f}_i}\cdot(-1)^{\scriptsize \bx\T\bS_r{\bf f}_i}.
\end{equation}
Note that above only the last term depends on $\bx$. 
Now if we multiply \eqref{e-shift} with the Hadamard matrix \eqref{e-Hadamard} we obtain
\begin{equation}\label{e-shift1}
  i^{\scriptsize {\bf f}\T_{\!\!i}\bS_r{\bf f}_i}\cdot(-1)^{\scriptsize {\bf b}\T_{\!\!r}{\bf f}_i}\sum_{\scriptsize \bx\in \F_2^r}(-1)^{\scriptsize \bx\T(\bv+\bS_r{\bf f}_i)},
\end{equation}
for all $\bx\in \F_2^r$ (where we have omitted the scaling factor). Then \eqref{e-shift1} is nonzero precisely when $\bv = \bS_r{\bf f}_i$ - the $i$th column of $\bS_r$. With $\bS_r$ in hand, one recovers ${\bf b}_r$ similarly by multiplying $\wb(\bx)\ov{\bw_{\scriptsize {\bf 0}}(\bx)}$ with the Hadamard matrix. To recover ${\bf b}_{m-r}$ one simply uses the knowledge of nonzero coordinates and \eqref{e-p1a1}. Next, with ${\bf b}$ in hand and the knowledge of the on-off pattern one recovers $\bH_\cI$ (and thus $H$) using \eqref{e-bmr} or equivalently \eqref{e-bmr1}. We will refer to the process of finding the column index ${\bf b}$ as \emph{dechirping}.

In the above somewhat ad-hoc method we did not take advantage of the geometric structure of the subspace chirps as eigenvectors of given maximal stabilizers or equivalently as the columns of given Clifford matrices. We do this next by following the line of \cite{TC19}.

Let $\bw$ be a subspace chirp as in \eqref{e-bssc}, and recall that it is the column of $\bG:=\bG_{\scriptsize{\bF}} = \bG_D(\bP\T) \bG_U(\bS)\bG_{\mathbf{\Omega}}(r)$ where $\bF:=\bF_{\mathbf{\Omega}}(r)\FUS\bF_D(\bP\T)$. Then by construction $\bG$ and $\bF$ satisfy $\bG^\dagger\bE(\bc)\bG = \pm\bE(\bc\T\bF^{-1})$ for all $\bc\in \F_2^{2m}$. Recall also from Theorem \ref{T-ES} that $\bG$ is the common eigenspace of the maximal stabilizer 
\begin{equation}\label{e-ES1}
    \bE(\Imr\bP\T,\,(\Imr\bS+\Imrr)\bP^{-1}\!) = \bE\!\left(\!\left[\!\begin{array}{c|c} \!\bH_{\cI}\!\!\!^{\T}\!&\!\bS_r\bI_{\cI}^{\T}\!\! \\ \!{\bf 0}&\!\widetilde{\bH_{\cI}}^{\sf T}\!\!\end{array} \!\right] \!\right).
\end{equation}
Thus, to reconstruct the unknown subspace chirp $\bw$, it is sufficient to first identify the maximal stabilizer that stabilizes it, and then identify $\bw$ as a column of $\bG$. The best way to accomplish the latter task, dechirping that is, is as described above, and thus we focus only on the former task. 
A crucial observation at this stage is the fact that the maximal stabilizer in \eqref{e-ES1} has precisely $2^r$ off-diagonal and $2^{m-r}$ diagonal Pauli matrices; see also Corollary~\ref{C-SZN}.

We now make use of the argument in Theorem \ref{T-ES}, that is, $\bw$ is an eigenvalue of $\bE(\bc)$ iff $\bE(\bc\T\bF^{-1})$ is diagonal. 
Let us focus first on identifying the $2^{m-r}$ diagonal Pauli matrices that stabilize $\bw$, that is, $\bc = \left[\!\!\begin{array}{c} \mathbf{0}\\\by \end{array}\!\!\right]$.
First we see that
\begin{equation}\label{e-FI}
    \bF^{-1} = \left[\!\!\begin{array}{cccc} \bI_{\cI}\bS_r & \widetilde{\bH_{\cI}} & \bI_{\cI} & {\bf 0} \\ \bH_{\cI} & {\bf 0}&{\bf 0}&\bI_{\widetilde{\cI}}\end{array}\!\!\right].
\end{equation}
Then for such $\bc$, $\bw$ is an eigenvector of $\bE(\bc)$ iff $\by\T\bH_{\cI} = 0$ iff $\by = \widetilde{\bH_{\cI}}\bz$ for some $\bz\in\F_2^{m-r}$. 
Thus, to identify the diagonal Pauli matrices that stabilize $\bw$, and consequently the subspaces $\bH_{\cI}, \widetilde{\bH_{\cI}}$, it is sufficient to find $2^{m-r}$ vectors $\by\in \F_2^{m}$ such that 
\begin{equation}\label{e-abs}
   0 \neq \bw^\dagger \bE({\bf 0},\by)\bw = \bw^\dagger\bE({\bf 0}, \wHI\bz)\bw.
\end{equation}
It follows by \eqref{e-Eab1} that the above is equivalent with finding $2^{m-r}$ vectors $\by$ such that 
\begin{equation}
   0 \neq \sum_{\scriptsize \bv\in\F_2^m}(-1)^{\scriptsize \by\T\bv}|\bw(\bv)|^2 = \sum_{\scriptsize \bv\in\F_2^m}(-1)^{\scriptsize \bz\T\widetilde{\bH}\T\bv}|\bw(\bv)|^2.
\end{equation}
The above is just a Hadamard transform which can be efficiently undone.

With a similar argument, $\bw$ is an eigenvector of a general Pauli matrix $\bE(\bx,\by)$ iff 
\begin{equation}\label{e-shift2}
  \bw^\dagger\bE(\bx,\by)\bw = i^{\scriptsize \bx\T\by}\sum_{\scriptsize \bv\in\F_2^m}(-1)^{\scriptsize \bv\T\by}\ov{\bw(\bv+\bx)}\bw(\bv) \neq 0.
\end{equation}
The above is again just a Hadamard transform. 
In fact, we see here both the ``shift'' (by $\bx$), the ``multiply'', and the Hadamard transform of the ``shift and multiply''. This is the main insight that transfers the shift and multiply technique of~\cite{HCS08} to computation with Pauli matrices. 
By definition, the Pauli matrix $\bE(\bx,\by)$ has a diagonal part determined by $\by$ and an off-diagonal part determined by $\bx$. 
The off-diagonal part of a Pauli determines the shift of coordinates whereas the diagonal part takes care of the rest.

Computing $\bw^\dagger \bU\bw$ for a generic $N\times N$ matrix is expensive, and even more so if the same computation is repeated $N^2$ times. 
However, when $\bU$ is a Pauli matrix, which is a monomial matrix of sparsity/rank 1, the same computation is much faster. 
Moreover, as we will see, for a rank $r$ BSSC we need not compute all the possible $N$ shifts but only $r$ of them.
This is an intuitive observation based on the shape of the maximal stabilizer \eqref{e-ES1}.
Indeed, once the diagonal Pauli matrices are identified, one can use that information to search the off-diagonal Pauli matrices only for $\bx\in \cs(\HI)$, which reduces the search from $2^m$ to $2^r$.
In fact, as we will see, instead of $2^r$ shifts we will need only use the $r$ shifts determined by columns of $\HI$.

Let us now explicitly make use of \eqref{e-shift2} to reconstruct the symmetric matrix $\bS_r$, while assuming that we have already reconstructed $\bH_{\cI}, \widetilde{\bH_{\cI}}$. 
In this case, as we see from~\eqref{e-FI}, the only missing piece of the puzzle is the upper-left block of $\bF^{-1}$.
We proceed as follows. For $\bc = \left[\!\!\begin{array}{c} \bx\\\by \end{array}\!\!\right]$, we have $\bw^\dagger\bE(\bx,\by)\bw\neq 0$ iff $\bE(\bc\T\bF^{-1})$ is diagonal, iff
\begin{equation}\label{e-SrRec}
    \bx\T[\bI_{\cI}\bS_r \,\,\,\, \widetilde{\bH_{\cI}}] = \by\T[\bH_{\cI}\,\,\,\,{\bf 0}].
\end{equation}
As before, we are interested in $\by\in\F_2^m$ that satisfy \eqref{e-SrRec}. 
First note that solutions to \eqref{e-SrRec} exist only if $\bx\T\widetilde{\bH_{\cI}} = {\bf 0}$, that is only if $\bx = \bH_{\cI}\bz$, $\bz\in \F_2^r$.
For such $\bx$, making use of \eqref{e-PI1}, we conclude that \eqref{e-SrRec} holds iff 
\begin{equation}\label{e-SrRec1}
    \bz\T\bS_r = \by\T\bH_{\cI}.
\end{equation}
Solutions of \eqref{e-SrRec1} are given by
\begin{equation}\label{e-SrRec2}
    \by = \widetilde{\bH_{\cI}}\bv + \bI_{\cI}\bS_r\bz,\,\,\,\,\bv\in\F_2^{m-r}.
\end{equation}
If we take $\bz = {\bf f}_i$ - the $i$th standard basis vector of $\F_2^r$ - we have that $\bz\T\bS_r$ is the $i$th row/column of $\bS_r$ while $\bx = \bH_{\cI}\bz$ is the $i$th column of $\bH_{\cI}$.

We collect all these observations in Algorithm~\ref{alg}.

\begin{algorithm}\caption{Reconstruction of single noiseless BSSC}\label{alg}
{\bf Input:} Unknown BSSC $\bw$
\begin{algorithmic}
\STATE~1. Compute $\bw^\dagger \bE({\bf 0},\by)\bw$ for $\by\in \F_2^m$.
\STATE~2. Find $\bH_{\cI}$ using 
\[
    \bw^\dagger \bE({\bf 0},\by)\bw \neq 0 \text{ iff }\by\T\bH_{\cI} = {\bf 0} \text{ iff } \by\in \cs(\widetilde{\bH_{\cI}}).
\]
\vspace{-0.2 in}
\STATE~3. Construct $\PI$ as in~\eqref{e-PI}.
\STATE~4. $r = \rk(\bH_{\cI})$. 
\STATE~5. $\bb{for}$ $i = 1,\ldots,r$ do:
\STATE~6. \quad Compute $\bw^\dagger \bE(\bH_{\cI}{\bf f}_i,\by)\bw$ for $\by\in\F_2^m$.
\STATE~7. \quad Determine the $i$th row of $\bS_r$ using \eqref{e-SrRec2}.
\STATE~8. {\bf end for}
\STATE~9. Dechirp $\bw$ to find ${\bf b}$.
\end{algorithmic}
{\bf Output:} $r,\bS_r,\PI,{\bf b}$.
\end{algorithm}

\subsection{Reconstruction of Single BSSC in the Presence of Noise}
In order to move towards a multi-user random access scenario, one needs a reliable reconstruction algorithm of noisy BSSCs. For this we consider a signal model
\begin{equation}
    \bs = \bw + \bn,
\end{equation}
where $\bn$ is Additive White Gaussian Noise (AWGN).
In such instance, the subspace reconstruction, that is, step (2) of Algorithm~\ref{alg} is a delicate procedure. However, one can proceed as follows. For each $\by\in\F_2^m$ we compute $\bs^\dagger\bE({\bf 0},\by)\bs$ and use it as an estimate of $\bw^\dagger\bE({\bf 0},\by)\bw$. 
We sort these scattered real values in decreasing order and make rank hypothesis, that is, for each $0\leq r \leq m$ we select $2^{m-r}$ largest values and proceed with Algorithm~\ref{alg} to obtain $\bw_r$. 
We then select the best rank using the Euclidean norm, that is,
\begin{equation}
    \widetilde{\bw} = \arg\min_r \|\bs - \bw_r\|_2.
\end{equation}

\begin{figure}[ht!]
\begin{center}
\includegraphics[width = 0.5\textwidth]{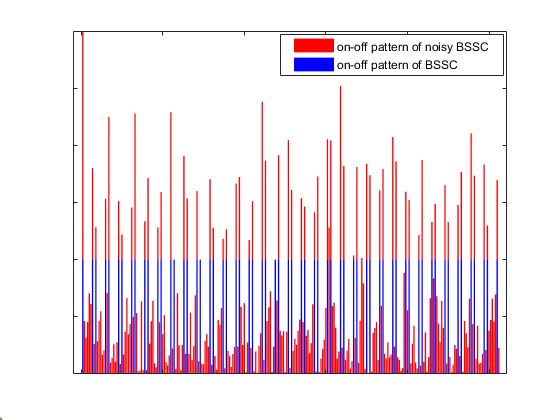}
\caption{On-off pattern of noisy BSSC versus on-off pattern of noiseless BSSC.} 
\label{fig1}
\end{center}
\end{figure}

In Figure~\ref{fig1} we see an instance of a rank $r=2$ BSSC on-off pattern in $N = 2^8$ dimensions, with and without noise. 
In this case $\bw^\dagger\bE({\bf 0},\by)\bw$ is non-zero $2^{m-r} = 64$ times. 
In this instance, only 94\% of the 64 highest $\bs^\dagger\bE({\bf 0},\by)\bs$ values of the noisy version match the on-off pattern of $\bw$. 
However, this can be overcame in recusntruction by using the fact that the on-off pattern is determined by a subspace. 
Thus one can build up $\widetilde{\bH_{\cI}}$ in a greedy manner by starting with the highest values and then including linear combinations. 
This strategy was tested in~\cite{TC19} with Monte-Carlo simulations
yielding low error rates even for low Signal-to-Noise Ratio (SNR); see~\cite[Fig.~1]{TC19}. There it was observed that, rather remarkably, BSSCs outperform BCs despite having the same minimum distance.



\section{Multi-BSSC Reconstruction}
\label{Sec-multiBSSC}

The strategy of noisy single BSSC reconstruction can be used as a guideline to generalize Algorithm~\ref{alg} to decode multiple simultaneous transmissions in a block fading multi-user scenario 
\begin{equation}\label{e-linBSSC}
    \bs = \sum_{\ell = 1}^L h_\ell\bw_\ell + \mathbf{n}.
\end{equation}
Here the channel coefficients $h_\ell$ are $\cC\cN(0,1)$, with neither
phase nor amplitude known, and $\bw_\ell$ are BSSCs. Noise
$\mathbf{n}$ may be added, depending on the scenario. This model
represents, e.g., a random access scenario, where $L$ randomly chosen
active users transmit a signature sequence, and the receiver should
identify the active users. In such application, the channel gain is
not known at the receiver, and thus one cannot use the amplitude to
transmit information. For this reason, the amplitude/norm is assumed,
without loss of generality, to be one. Additionally, the channel phase
is also not known at the receiver and should not carry any
information. Thus without loss of generality, the codewords can be
assumed to come from a Grassmannian codebook, such as $\chirp$ or
$\bssc$.

We generalize the single-user 
algorithm to a multi-user algorithm, where the coefficients $h_\ell$ are estimated to identify the most probable transmitted signals. For this, we use Orthogonal Matching Pursuit (OMP), which is analogous with the strategy of~\cite{HCS08}. We assume that we know $L$. 
\begin{algorithm}\caption{Reconstruction of noiseless multi-BSSCs}\label{alg1}
{\bf Input:} Signal~$\bs$ as in~\eqref{e-linBSSC}.
\begin{algorithmic}
\STATE~1. \textbf{for} $\ell=1:L$ do
\STATE~2. \quad \textbf{for} $r = 0:m$ do
\STATE~3. \qquad Greedily construct the $m-r$ dimensional subspace \\ 
\hspace{.425 in} $\wHI$ using the highest values of $|\bs^\dagger\bE({\bf 0},\by)\bs|$.
\STATE~4. \qquad Estimate $\widetilde{\bw}_r$ as in Alg.~\ref{alg}.
\STATE~5. \quad \textbf{end for} 
\STATE~6. \quad Select the best estimate $\widetilde{\bw}_\ell$.
\STATE~7. \quad Determine $\widetilde{h}_1,\ldots,\widetilde{h}_\ell$ that minimize
$$\left\|\bs - \sum_{j = 1}^\ell h_j\widetilde{\bw}_j\right\|_2.$$
\STATE~8. \quad Reduce $\bs$ to $\bs' = \bs - 
 \sum_{j = 1}^\ell \widetilde{h}_j\widetilde{\bw}_j$.
 \STATE~9. \textbf{end for}
\end{algorithmic}
\textbf{Output: }$\widetilde{\bw}_1,\ldots, \widetilde{\bw}_L$.
\end{algorithm}
\begin{figure}
\centering
\includegraphics[width = 0.5\textwidth]{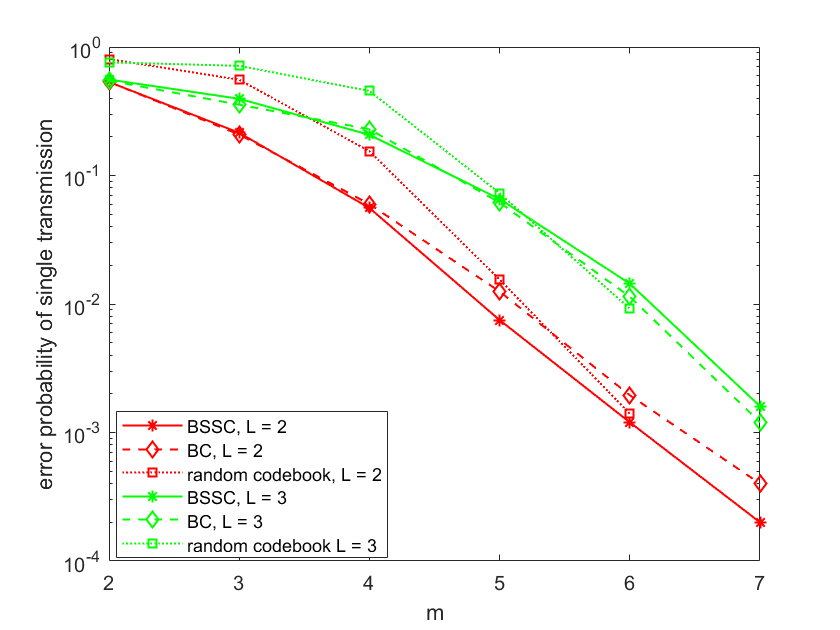}
\caption{Error probability of Algorithm~\ref{alg1} on absence of noise. Random codebook included for comparison.} 
\label{fig}
\end{figure}

The estimated error probability of single user transmission for $L = 2,3$ is given in Figure~\ref{fig}. For the simulation, the rank $r$ is selected in a weighted manner, according to the relative size of rank $r$ BSSCs (recall that there are $2^m\cdot{m\choose r}_2\cdot 2^{r(r+1)/2}$ rank $r$ BSSCs). 
Whereas, within a given rank, BSSCs are chosen uniformly. 
We compare the results with BC codebooks and random codebooks with the same cardinality. For random codebooks, steps (2)-(5) of Algorithm~\ref{alg1} are substituted with exhaustive search (which is infeasible is beyond $m = 6$).

\begin{figure}
\begin{center}
\includegraphics[width = 0.5\textwidth]{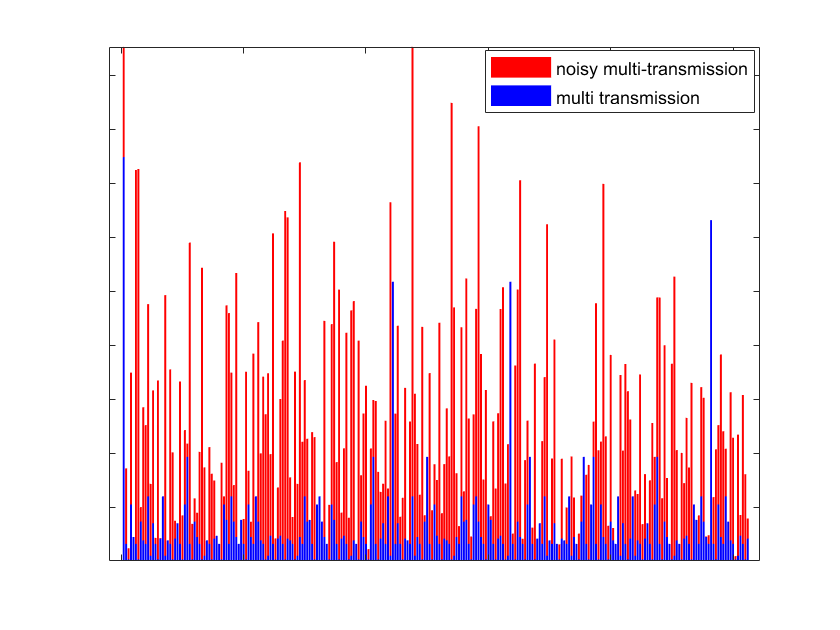}
\caption{On-off pattern of a noiseless vs. noisy linear combination of BSSCs.} 
\label{mBSSC}
\end{center}
\end{figure}

The erroneous reconstructions of Algorithm~\ref{alg1} come in part from steps (3)-(4). Specifically, from the cross-terms of 
\begin{equation*}
    \bs^\dagger\bs = \sum_{\ell = 1}^L|h_\ell|^2\|\bw_\ell\|^2 + \sum_{i\neq \ell} \ov{h_i}h_\ell\bw_i^\dagger\bw_\ell.
\end{equation*}
For BCs, these cross-terms are the well-behaved \emph{second order Reed-Muller functions}. 
The BSSCs, unlike the BCs~\cite{CHJ10}, do not form a group under point-wise multiplication (Theorem~\ref{T-SemiGroup}), and thus the products $\bw_i^\dagger\bw_\ell$ are more complicated. 
Indeed, when two BSSCs of different ranks and/or different on-off patterns are multiplied coordinate-wise (which we do during the ``shift and multiply'') the resulting BSSCs could be potentially very different (if not zero) as described in Theorem~\ref{T-SemiGroup}.
In addition, linear combinations of BSSCs~\eqref{e-linBSSC} may perturb the on-off patterns of the constituents, and depending on the nature of the channel coefficients $h_\ell$, the algorithm may detect a higher rank BSSC in $\bs$. 
If the channel coefficients of two 
BSSCs happen to have similar amplitudes, the algorithm may detect a lower rank BSSC that corresponds to the overlap of the on-off patterns of the BSSCs. These phenomena are depicted in Figure~\ref{mBSSC} (in blue) in which the on-off pattern of a linear combination of a rank two, a rank three, and a rank six BSSCs in $N = 2^8$ dimensions is displayed.
There, we see multiple levels (in blue) of $\bs^\dagger\bE({\bf 0},\by)\bs$, only some of which correspond to actual on-off patterns $\bw_\ell^\dagger\bE({\bf 0},\by)\bw_\ell$ of the given BSSCs, and the rest corresponds to different combinations of overlaps. The problems for multi-BSSC reconstruction caused by these phenomena is alleviated by the fact that most BSSC codewrods have high rank. E.g., as $m$ grows, it follows by Corollary~\ref{C-size} that about $42\%$ of BSSCs are BCs. Low rank BSSCs are very unlikely in~\eqref{e-linBSSC}.

Despite these phenomena affecting BSSC on-off patterns in multi-BSSC
scenarios, a
decoding algorithm like the one discussed is able to distinguish
different levels and provide reliable performance. It is worth
mentioning that by comparing Figure~\ref{fig} with~\cite[Fig.~1]{TC19}
we see that the interference of BSSCs is much more benign than general
AWGN, which in turn explains the reliable reconstruction of noiseless
multi-user transmission.

Interestingly, even in this multi-user scenario, we see that BSSCs
outperform BCs.
With increasing $m$, the performance benefit of the algebraically
defined codebook over random codebooks diminishes. However, the
decoding complexity remains manageable for the algebraic codebooks.

In~\cite{TC19} it was demonstrated that reconstruction of a single noisy
BSSC was possible even for low SNR. We have performed preliminary
simulations and have tested Algorithm~\ref{alg1} on a noisy multi-user
transmission. Unlike the single BSSC scenario, the multi BSSCs
scenario requires a higher SNR regime for reliable performance. In
Figure~\ref{mBSSC} we have shown (in red) $|\bs^\dagger\bE({\bf
  0},\by)\bs|$ for a noisy version of the same linear combination as
before (displayed in blue). In this instance we have fixed SNR = 8 dB.
A close look shows that this scenario is
different from the single user scenario displayed in
Figure~\ref{fig1}. In this instance, even an exhaustive search over
ranks $r$ as in Algorithm~\ref{alg1} produces an on-off pattern that
matches at most $61\%$ \emph{any} actual on-off pattern, and thus the
subspace reconstruction inevitably fails. On the other hand, if the
on-off pattern is reconstructed correctly, then the corresponding
$r$-dimensional BC can be reconstructed reliably. When noise is on
manageable level, reliable reconstruction of multi-user BSSCs is
possible with Algorithm~\ref{alg1}. In Figure~\ref{mBC}, we depict the
performance of $N=256$ BSSC and BCs in a scenario with SNR 30 dB, for
a varying number of simultaneously transmitting users. Again, we see
that BSSCs provide slightly better error performance than BCs, despite
the codebook being larger.

\begin{figure}
\begin{center}
\includegraphics[width = 0.5\textwidth]{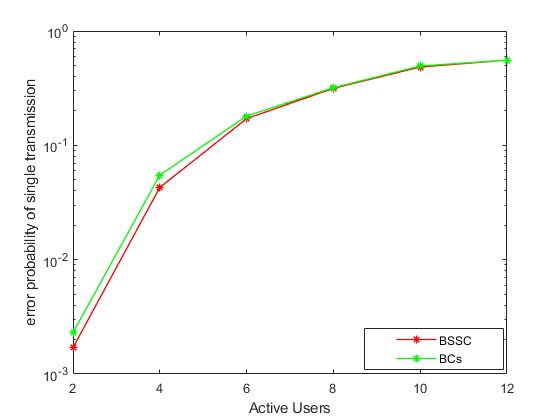}
\caption{Error probability of Algorithm~\ref{alg1} for noisy multi-user transmission in $N = 256$ dimensions and SNR = 30 dB.} 
\label{mBC}
\end{center}
\end{figure}


\section{Conclusions and Future Research}

Algebraic and geometric properties of BSSCs are described in details. 
BSSCs are characterized as common eigenspaces of maximal sets of commuting Pauli matrices, or equivalently, as columns of Clifford matrices. 
This enables us to fully exploit connections between symplectic geometry and quantum computation, which in turn yield considerable complexity reductions. 
Further, we have developed a low complexity decoding algorithm for multi BSSCs transmission with low error probability. 

By construction, BSSCs inherit all the desirable properties of BCs, while having a higher cardinality. In wireless communication scenarios BSSCs exhibit slightly lower error probability than BCs. For these reasons we think that BSSCs constitute good candidates for a variety of applications.

Algorithm~\ref{alg1} is a generalization of the BC decoding algorithm of~\cite{HCS08} to BSSCs. As pointed out in~\cite{Calderbank2019}, the decoding algorithm of~\cite{HCS08} does not scale well in a multi-user scenario, in terms of the number of users supported as a function of codeword length. 
In~\cite{Calderbank2019,massiveGuo} slotting arrangements were added on top of BC codes to increase the length, and the number of supported users. Part of the information in a transmission is embedded in the choice of a BC, part in the choice of active slots.
In~\cite{Calderbank2019}, interference cancellation across slots is applied, and the discussed scheme can be considered a combination of physical layer (PHY) BC coding, and a Medium Access Control (MAC) Layer code of the type discussed in~\cite{Liva2011}.  The works of~\cite{Calderbank2019,massiveGuo} show that following such principles, practically implementable massive random access schemes, operating in the regime of interest of~\cite{Polyanskiy17}, can be designed. If the small-$m$ BC-transmissions in the slots would be replaced with BSSC transmissions with the same $m$, the results of this paper indicate that performance per slot would be the same, if not slightly better than in~\cite{Calderbank2019,massiveGuo}. This indicates that combined MAC/PHY codes, where BSSC would be the PHY component instead of BC as used in~\cite{Calderbank2019,massiveGuo}, are likely to provide slightly higher rates with otherwise similar performance as~\cite{Calderbank2019,massiveGuo}. In future work, we plan to investigate such codes.

As mentioned, we have seen in all our simulations that BSSCs outperform BCs. Although our algorithms do not find the closest codeword, this may be due to a fact that BSSCs have fewer closest neighbors on average than BCs. We will investigate this in future work with a statistical analysis of Algorithm~\ref{alg1} along the lines of~\cite{CHJ10}.

Binary chirps have been generalized in various works to prime dimensions, and recently to non-prime dimensions~\cite{Pitaval20}. In future work we will consider analogues generalizations of BSSCs, by adding a sparsity component to generalized BCs and/or by lifting BSSCs modulo $2t$. 

As a byproduct, we have obtained a Bruhat decomposition of the symplectic group that involves five elementary symplectic matrices (compared to the seven layers of~\cite{Can18}, c.f., \eqref{e-trung}). We think that this has implications in quantum computation. In future research we will explore whether Algorithm~\ref{algo:alg1} can be leveraged to improve upon~\cite{MR18,Koenig-jmp14}.


\section*{Acknowledgements}
The work of TP and OT was funded in part by the Academy of Finland (grant 319484). The work of RC was supported in part by the Air Force Office of Scientific Research (grant FA 8750-20-2-0504). The authors would like to thank Narayanan Rengaswamy for helpful discussions.

\bibliographystyle{IEEEtran}
\bibliography{IEEEabrv,BSSC}

\end{document}